\documentclass[a4paper,11pt]{article}
\usepackage{bm,mathrsfs,amsmath,amsthm,amssymb,eepic,amscd,longtable,array,graphicx}
\usepackage[dvips]{color}
\newcommand{\nc}{\newcommand}
\nc{\tcr}{\textcolor{red}}
\nc{\rnc}{\renewcommand}
\nc{\nn}{\nonumber}
\nc{\der}{{\partial}}
\rnc{\Im}{{\rm{Im}\,}}
\rnc{\Re}{{\rm{Re}\,}}
\nc{\db}{\displaybreak[0]\\}
\nc{\bra}{\langle}
\nc{\ket}{\rangle}
\nc{\bs}{\boldsymbol}

\DeclareMathOperator{\Tr}{Tr}

\DeclareMathOperator{\End}{End}

\newtheorem{theorem}{Theorem}[section]
\newtheorem{lemma}[theorem]{Lemma}

\newtheorem{proposition}[theorem]{Proposition}
\newtheorem{corollary}[theorem]{Corollary}

\theoremstyle{definition}
\newtheorem{definition}[theorem]{Definition}
\newtheorem{example}[theorem]{Example}

\numberwithin{equation}{section}

\numberwithin{equation}{section}

\begin{document}%
%
\title{
$K$-theoretic boson-fermion correspondence and melting crystals}

\author{
Kohei Motegi$^1$\thanks{E-mail: motegi@gokutan.c.u-tokyo.ac.jp} \,
and
Kazumitsu Sakai$^2$\thanks{E-mail: sakai@gokutan.c.u-tokyo.ac.jp}
\\\\
$^1${\it Okayama Institute for Quantum Physics, }\\
 {\it Kyoyama 1-9-1, Okayama 700-0015, Japan} \\
\\
$^2${\it Institute of physics, University of Tokyo,} \\ 
{\it Komaba 3-8-1, Meguro-ku, Tokyo 153-8902, Japan}
\\\\
\\
}

\date{\today}

\maketitle

\begin{abstract}
We study non-Hermitian integrable fermion and boson systems 
from the perspectives of Grothendieck polynomials. The models 
considered in this article are the five-vertex model as a 
fermion system and the non-Hermitian phase model as a 
boson system. Both of the models are characterized by the 
different solutions satisfying the same Yang-Baxter relation. From our 
previous works on the identification between the wavefunctions 
of the five-vertex model and Grothendieck polynomials, 
we introduce skew Grothendieck polynomials, and derive the 
addition theorem among them. Using these relations,  we derive
the wavefunctions of the non-Hermitian phase model as a determinant
form which can  also be expressed as the Grothendieck polynomials. 
Namely, we establish a $K$-theoretic boson-fermion correspondence 
at the level of wavefunctions. As a by-product, 
the partition function of the statistical mechanical model of a 3D
melting crystal is exactly calculated by use of the scalar
products of the wavefunctions of the phase model. The resultant
expression can be regarded as a $K$-theoretic generalization of the  
MacMahon function describing the generating function of the plane 
partitions, which interpolates the generating functions of 
two-dimensional and three-dimensional Young diagrams.
\end{abstract}

\section{Introduction}
%
The symmetric polynomials is the basic object in representation theory, 
combinatorics and related geometry. It also appears in mathematical 
physics, especially in the integrable models.
The most fundamental symmetric polynomial is the Schur polynomials,
which appears as the solutions of the KP hierarchy \cite{DJKM} and
wavefunctions of the phase model \cite{Bo,SU,KS,BBF,Wh,OY} for example.
It can also be used to construct a determinantal process
named the Schur process \cite{OR} which have
applications to the partition functions of the topological strings \cite{ORV},
for example.

We recently extended the relation between
the Schur polynomials and the integrable models,
and found that the wavefunctions of the one-parameter
family of the integrable five-vertex models
can be represented in Grothendieck polynomials \cite{MS}.
The Grothendieck polynomials
was originally introduced in the context of algebraic geometry
\cite{LS,FK,IN,IS}
as structure sheaf of the Schubert variety
in the $K$-theory of flag varieties.
By the identification of the wavefunctions
with the Grothendieck polynomials for Grassmannian varieties,
the determinant representations of the scalar products,
which is the inner product between the wavefunctions,
is nothing but the Cauchy identity for the Grothendieck polynomials.
We also revealed the meaning of the orthogonality to show that
the Grothendieck polynomials is a discrete orthogonal polynomial over the
``Cassini oval" \cite{GM}, the solution curve of the Bethe equations.

The integrable five-vertex model is related to the non-Hermitian quantum 
integrable spin chain and the stochastic process called the
totally asymmetric simple exclusion process (TASEP) \cite{MGP}. The TASEP is 
a many-particle stochastic process with exclusion as an interaction, which 
can be viewed as a natural generalization of random walk. From the following 
perspectives, these models can be regarded  as {\it fermion} systems. 
First, the space on which the Hamiltonian or the stochastic matrix acts 
is the tensor product of copies of two-dimensional space spanned by the empty
state and particle-occupied state, i.e. the double occupancy is forbidden.
Second, the above models are in one-to-one correspondence with fermion
systems through the Jordan-Wigner transformation. Finally for the
interaction-free case, the physical quantities, such as wavefunctions,
and the number of configurations of stochastic particles, etc.
are represented as the Schur polynomials which can be described
in terms of the formalism of the fermion and its Fock space.

In this paper, we study another type of integrable lattice model
derived by a different solution satisfying the same Yang-Baxter relation
for the five-vertex model. The model discussed in this paper is a 
boson model called the non-Hermitian phase model \cite{BN}, which is a 
one-parameter generalization of the phase model \cite{BIK}.
At a special point of the parameter, the non-Hermitian phase model
describes the totally asymmetric zero range process (TAZRP),
i.e., a stochastic process for a system of bosons which,
in contrast to the TASEP, the particles are allowed to occupy the same site.
The wavefunctions of the phase model was shown to be expressed 
as the Schur polynomials \cite{Bo}. In this sense, the phase model
can be interpreted as the free fermion systems.
We show that the one-parameter family of the phase model
corresponds to the generalization from the Schur polynomials
to the Grothendieck polynomials.
Namely, we show that the wavefunctions of the non-Hermitian phase model
is nothing but the Grothendieck polynomials: we establish a $K$-theoretic
boson-fermion correspondence at the level of the wavefunctions.
We show this by introducing the skew Grothendieck polynomials
and by deriving an addition theorem satisfied by the skew Grothendieck
polynomials. The skew Grothendieck polynomials can be introduced
in the context of the integrable five-vertex model naturally
from the relation between the wavefunctions of the $N$-particle state
and the $N$-variable Grothendieck polynomials (see \cite{Mc,IS} for another 
definition introduced from perspectives of combinatorics).
By this boson-fermion correspondence, the determinant representations 
of the scalar products and the summation of the wavefunctions
follow from the Cauchy identity and the summation formula
for the Grothendieck polynomials.
The Cauchy identity \cite{MS} used in this paper
is different from the dual Cauchy identity \cite{LN}
which is the pairing between Grothendieck polynomials and dual Grothendieck
polynomials.
Our approach is based on the quantum inverse scattering method which starts
from the $L$-operator.
There is another approach to the wavefunction from the coordinate Bethe ansatz
\cite{Take}, where the equivalence with the Grothendieck polynomials follows
as a consequence.

As another application of the above mentioned boson-fermion correspondence, 
we study the statistical mechanical model of a three-dimensional 
melting crystal. The model is in one-to-one correspondence with 
the plane partitions which  is regarded as a three-dimensional
extension of the Young diagrams. The partition function of the
model becomes a generating function of the plane partitions.
We show the partition function can be exactly calculated by
the scalar product of the non-Hermitian phase model. For the
finite volume, the partition function can be given by a
determinant form which reproduces MacMahon's generating function \cite{Mac}
at a special point of the parameter. In the infinite volume limit, 
the partition function is explicitly given by an infinite
product which is regarded as a $K$-theoretic generalization of the  
MacMahon function \cite{Mac}.
The $K$-theoretic MacMahon function
interpolates the ordinary MacMahon function and 
Euler's generating function of partitions.
Namely, it unifies the generating functions of the two-dimensional
and three-dimensional Young diagrams.
Note that there are other types of three-dimensional melting crystal models
\cite{DZ,FZ,Vu,FW} whose constructions are based on connections with
integrable models such as the loop models related to the
XXZ chain at roots of unity and 
free fermion models, or connections with symmetric polynomials 
such as the Schur polynomials and its generalization to the 
Hall-Littlewood and the Macdonald polynomials. 
Our model is different from them, and is based on the non-Hermitian integrable
spin chain and phase model,
whose wavefunctions are the Grothendieck polynomials.
The directions of extending the Schur polynomials to the 
Grothendieck and Macdonald polynomials are different,
hence the explicit forms of the corresponding skew polynomials and the weights
assigned to each plane partition are totally different between the one
in this paper and the ones in previous literature.
The Hall-Littlewood polynomials have representations in terms
of vertex operators, and many properties including the connection with the
melting crystal model can be treated in the same way for the Schur polynomials.
However, there is no such vertex operator representation
for the Grothendieck polynomials, and we approach to the problem 
of construction by using the correspondence
with the non-Hermitian integrable models.

This paper is organized as follows.
In the next section, we review the relation between the
wavefunctions of the integrable five-vertex model and the
Grothendieck polynomials. In section 3, we introduce the skew 
Grothendieck polynomials and derive an addition theorem 
satisfied by them.
In section 4, we introduce the non-Hermitian phase model,
and show that the wavefunctions can be expressed as Grothendieck polynomials
in section 5. In section 6, we discuss the melting crystal and
derive the exact expressions of the partition function of the model.
Section 7 is devoted to summary and discussion.
\section{Grothendieck polynomials and five-vertex models}
%
In this section, we recall a relationship between 
Grothendieck polynomials and the integrable five-vertex model 
\cite{MS}. Utilizing this relation, in the next section we introduce 
skew Grothendieck polynomials which play a key role in subsequent 
analysis.

Grothendieck polynomials were originally introduced as
polynomial representatives of structure sheaf of the Schubert variety
in the $K$-theory of flag varieties \cite{LS}.
The $\beta$-Grothendieck polynomials was introduced \cite{FK}
to unify the original Grothendieck polynomials and the Schubert polynomials,
which are structure sheaves for the $K$-theory ($\beta=-1$)
and the cohomology ($\beta=0$), respectively.
For the case when the flag variety is type $A$ Grassmannian varieties,
the Grothendieck polynomials can be represented as the
following determinant form \cite{IN}, which we regard  as the definition
of the Grothendieck polynomials.
\begin{definition} \cite{LS,FK,IN}
The Grothendieck polynomials is defined as the following determinant
\begin{align}
G_\lambda(z_1,\dots,z_N;\beta)=\frac{\mathrm{det}_N(z_j^{\lambda_k+N-k}
(1+\beta z_j)^{k-1})}{\prod_{1 \le j < k \le N}(z_j-z_k)}, \label{GR}
\end{align}
where $\{z_1,\dots,z_N\}$ is a set of variables and
$\lambda=(\lambda_1,\dots,\lambda_N)$
is a sequence of weakly decreasing
nonnegative integers $\lambda_1 \ge \dots \ge \lambda_N \ge 0$.
\end{definition}

Note that for the case of cohomology $\beta=0$,
the $\beta$-Grothendieck polynomials are nothing but the Schur polynomials,
which are Schubert polynomials for type $A$ grassmannian varieties.

In fact, the Grothendieck polynomials appear as wave-functions
in the five-vertex model. The five-vertex model is a two-dimensional
statistical mechanical model whose Boltzmann weights are given by 
the elements of the $L$-operator \cite{MS}
\footnote{The five-vertex model in this paper is different
from the one in \cite{BBF}.
The $R$-matrix in \cite{BBF} satisfying the $RLL$ relation is essentially
the trigonometric Felderhof model.
The $R$-matrix in this paper is a special limit of the XXZ chain
(the signs of weights when all spins are up and all spins are down
are different for the trigonometric Felderhof model,
and are the same for the XXZ chain),
and the corresponding $L$-operators are different.
For example, the configurations of the five-vertex models
having nonzero weights are different,
and the model in \cite{BBF} cannot create
either the $N$-particle state or its dual.
The model in this paper can create both the $N$-particle state
and its dual.}
$L_{a j}(u)\in\End(W_a \otimes V_j)$:
\begin{align}
L_{a j}(u)=u s_a s_j+\sigma_a^- \sigma_j^+
                   +\sigma_a^+ \sigma_j^-+(-\beta^{-1} u-u^{-1})n_a s_j
                   -\beta^{-1} u n_a n_j \quad (u\in \mathbb{C}),
\label{loperator}
\end{align}
where $W_a=\mathbb{C}^2$ (resp. $V_j=\mathbb{C}^2$)
denotes the $a$th auxiliary  space (resp. $j$th quantum space) 
spanned by the empty state $|0\ket_{a}=\binom{1}{0}_{a}$ (resp. 
$|0\ket_{j}=\binom{1}{0}_{j}$)
and particle occupied state $|1\ket_a=\binom{0}{1}_a$ 
(resp. $|1\ket_j=\binom{0}{1}_j$). The parameter $\beta$ can be taken arbitrary
(the parameter $\alpha$ in \cite{MS} corresponds to $\beta$ as
$\alpha=-\beta^{-1}$).
See also Figure \ref{weight} for a pictorial description of the $L$-operator 
\eqref{loperator}. 
The above $L$-operator is given by a solution to the following
Yang-Baxter relation ($RLL$-relation):
\begin{align}
R_{ab}(u,v)L_{a j}(u)L_{b j}(v)=
L_{b j}(v)L_{a j}(u)R_{ab}(u,v)
\label{RLL}
\end{align}
holding in $\End(W_a \otimes W_b \otimes V_j)$ for arbitrary
$u,v \in \mathbb{C}$. Here the matrix $R_{ab}(u,v)\in \End(W_a \otimes W_b)$
is defined by
\begin{align}
R(u,v)
=
\begin{pmatrix}
f(v,u) & 0 & 0 & 0 \\
0 & 0 & g(v,u) & 0 \\
0 & g(v,u) & 1 & 0 \\
0 & 0 & 0 & f(v,u)
\end{pmatrix},
\,\,
f(v,u)=\frac{u^2}{u^2-v^2},  \,
g(v,u)=\frac{uv}{u^2-v^2}, 
\label{Rmatrix}
\end{align}
which is a solution to the Yang-Baxter equation:
\begin{align}
R_{ab}(u,v)R_{ac}(u,w)R_{bc}(v,w)=
R_{bc}(v,w)R_{ac}(u,w)R_{ab}(u,v).
\label{YBE}
\end{align}

\begin{figure}[tt]
\begin{center}
\includegraphics[width=0.75\textwidth]{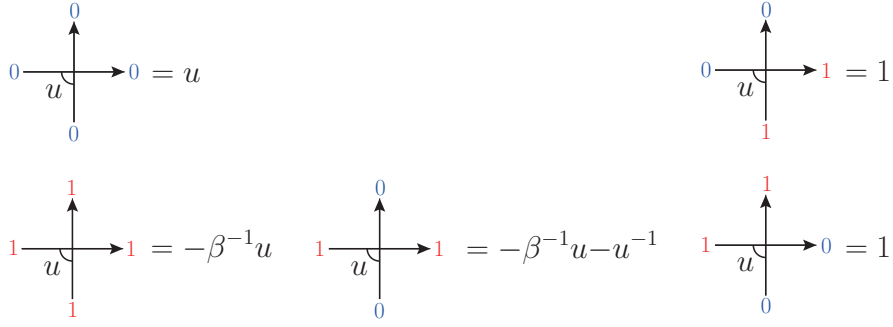}
\end{center}
\caption{The non-zero elements of the $L$-operator \eqref{loperator}.
The left (resp. up) arrow represents an auxiliary space 
(resp. a quantum space).
The indices 0 or 1 on the left (resp. right) of the vertices
denote the input (resp. output) states $|0\ket$ or $|1\ket$ 
in the auxiliary space, while those on the bottom (resp.  top) 
denote the input (resp. output) states in the quantum space. 
Note that the weights are invariant under a $180^{\circ}$ rotation.
}
\label{weight}
\end{figure}

Let us define the monodromy matrix $T(u)$ as a product of $L$-operators:
\begin{align}
T_{a}(u)=L_{a M}(u) \cdots L_{a 1}(u)&=
\begin{pmatrix}
A(u) & B(u)  \\
C(u) & D(u)
\end{pmatrix}
_{a}
\label{monodromy1}
\end{align}
acting  on $W_a \otimes (V_1\otimes\dots\otimes V_M)$.
Tracing out the auxiliary space, one obtains the transfer matrix 
$t(u)\in \End(V^{\otimes M})$
\begin{align}
t(u)=\Tr_{W_a} T_a(u)
\end{align}
 which commutes for different spectral parameters:
$[t(u),t(v)]=0$. The quantum Hamiltonian corresponding to the five-vertex
model is defined by $t(u)$:
\begin{align}
H:=
\sum_{j=1}^M 
\left\{
-\beta^{-1} \sigma_j^+\sigma_{j+1}^-+\frac{1}{4}
(\sigma_j^z\sigma_{j+1}^z-1)\right\}
=
\left.
\frac{\sqrt{-\beta}}{2}\frac{\der}{\der u} 
\log\left\{u^{-M} t(u)\right\}
\right|_{u=\sqrt{-\beta}}.
\label{Baxter}
\end{align}
Note that the above Hamiltonian, in general, is non-Hermitian.
For $\beta=-1$, the Hamiltonian corresponds to a stochastic
matrix describes a stochastic process called the totally asymmetric
simple exclusion process (TASEP).

The arbitrary $N$-particle state $|\psi(\{u \}_N) \ket$ 
(resp. its dual $\bra \psi(\{u \}_N)|$) 
(not normalized) with $N$ spectral parameters
$\{ u \}_N=\{ u_1,\dots,u_N \}$
is constructed by a multiple action
of $B$ (resp. $C$) operator on the vacuum state 
$|\Omega \ket:=| 0^{M} \ket:=|0\ket_1\
\otimes \dots \otimes |0\ket_{M}$
(resp. $\bra \Omega|:=\bra 0^{M}|:=
{}_1\bra 0|\otimes\dots \otimes{}_{M}\bra 0|$):
\begin{align}
|\psi(\{u \}_N) \ket=\prod_{j=1}^N B(u_j)| \Omega \ket,
\quad
\bra \psi(\{u \}_N)|=\bra \Omega| \prod_{j=1}^N C(u_j).
\label{statevector}
\end{align}
In \cite{MS}, we computed the overlap between the arbitrary 
off-shell\footnote{The terminology ``off-shell" means
that the set of parameters $\{u\}_N$ is arbitrary. On the
other hand ``on-shell" means that $\{u\}_N$ is taken so that
the $N$-particle state $|\psi(\{u \}_N)\ket$ is one of the
eigenstate of the Hamiltonian.}
$N$-particle state $|\psi(\{u\}_N)\ket$ and 
the (normalized) state with an arbitrary particle configuration 
$|x_1 \cdots x_N\ket$ $(x_1<\dots<x_N$), 
where $x_j$ denotes the positions of the particles. 
The wavefunction
$\bra x_1 \cdots x_N | \psi(\{u\}_N) \ket$ and its dual 
$\bra \psi(\{u\}_N)|x_1\cdots x_N \ket$ were found to be
given by the Grothendieck polynomials.
\begin{theorem} \label{th-wave} {\rm \cite{MS}}
The (off-shell) wavefunction and its dual wave-function
of the integrable five-vertex model
are, respectively, given by the Grothendieck polynomials as
\begin{align}
\bra x_1 \cdots x_N|\psi(\{ u \}_N) \ket&=(-\beta^{-1})^{N(N-1)/2}
\prod_{j=1}^N u_j^{M-1} G_\lambda(z_1,\dots,z_N;\beta), 
\label{wavefunctionone} \\
\bra \psi(\{ u \}_N)|x_1 \cdots x_N \ket&=(-\beta^{-1})^{N(N-1)/2}
\prod_{j=1}^N u_j^{M-1} G_{\lambda^\vee}(z_1,\dots,z_N;\beta),
\label{wavefunctiontwo}
\end{align}
where $z_j=-\beta^{-1}-u_j^{-2}$,  and
$\lambda=(\lambda_1,\dots,\lambda_N)$
($M-N \ge \lambda_1 \ge \cdots \ge \lambda_N \ge 0$)
and $\lambda^\vee=(\lambda_1^\vee,\dots,\lambda_N^\vee)$
($M-N \ge \lambda_1^\vee \ge \cdots \ge \lambda_N^\vee \ge 0$)
are the Young diagrams related to the particle configuration
$x=(x_1, \dots, x_N) $ as
$\lambda_j=x_{N-j+1}-N+j-1$ and $\lambda_j^\vee=M-N+j-x_j$,
respectively.
\end{theorem}
\noindent
Note that the Young diagram $\lambda^\vee$ is the complementary
part of the Young diagram $\lambda$ in the $N \times (M-N)$ rectangular
Young diagram.  Let $x_j^{\vee}$ be the particle configuration given by
$x_j^{\vee}=\lambda^{\vee}_{N-j+1}+j$.
The particle configurations $x=(x_1,\dots,x_N)$ 
($x_j=\lambda_{N-j+1}+j$) and $x^{\vee}=(x^{\vee}_1,\dots,x^{\vee}_N)$ 
for given $\lambda$ are connected by the relation:
\begin{align}
(x^{\vee}_1,\dots,x^{\vee}_N)=(M-x_N+1,\dots,M-x_1+1) 
\text{ for $x_j=\lambda_{N-j+1}+j$ and
 $x^{\vee}_j=\lambda^{\vee}_{N-j+1}+j$}.
\label{reverse}
\end{align}
In Figure~\ref{config}, we denote an example of 
$\lambda$ and $\lambda^{\vee}$ together with the corresponding
particle configurations $x_j=\lambda_{N-j+1}+j$ and 
$x^{\vee}_j=\lambda^{\vee}_{N-j+1}+j$. From this, one can intuitively 
find that the positions of the particles corresponding to $\lambda^{\vee}$ are 
related to those corresponding to $\lambda$ after a $180^{\circ}$ rotation.
\begin{figure}[ttt]
\begin{center}
\includegraphics[width=0.75\textwidth]{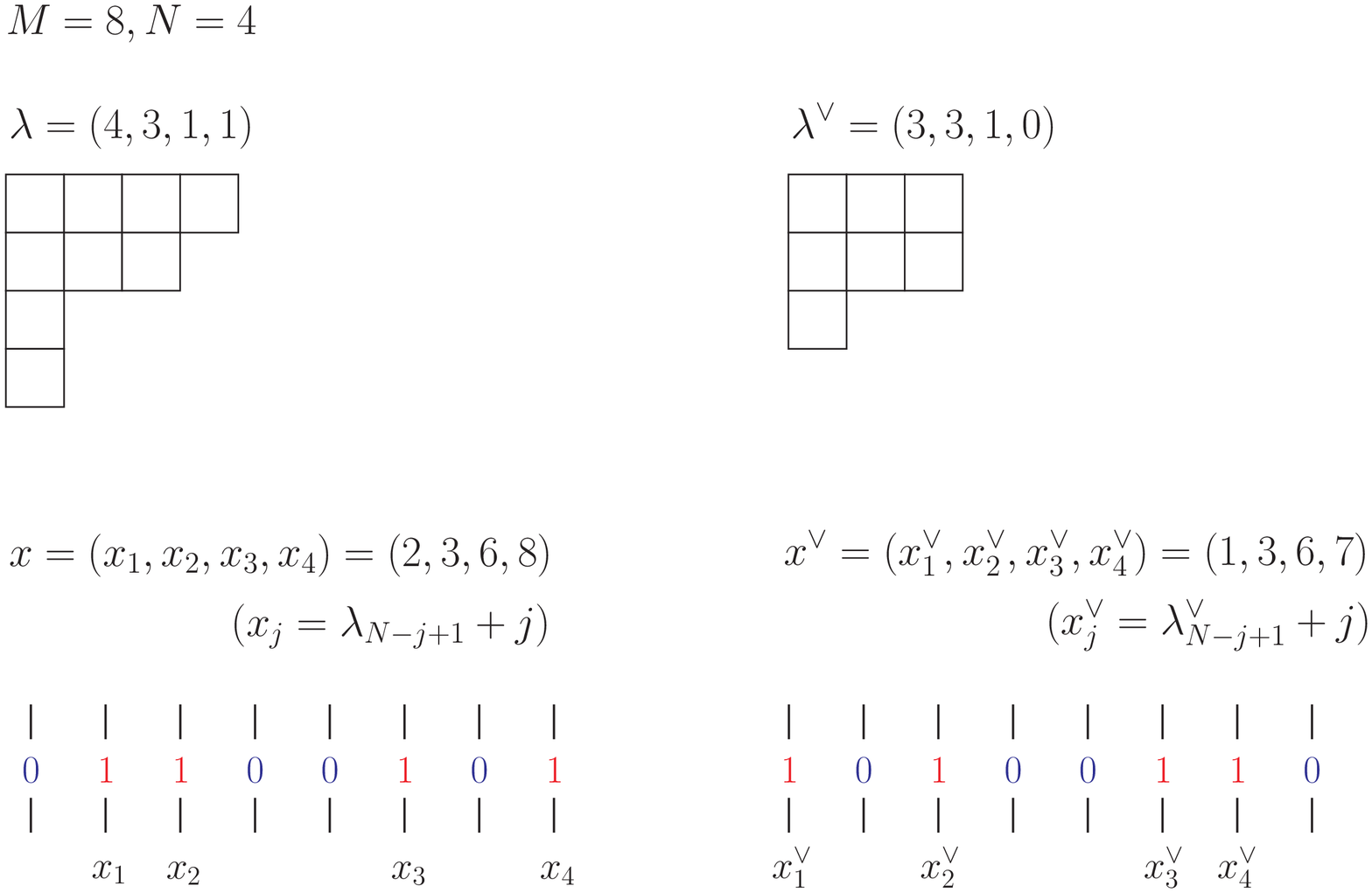}
\end{center}
\caption{An example of the partitions $\lambda$ and $\lambda^{\vee}$
and the corresponding particle configurations $x$ ($x_j=\lambda_{N-j+1}+j$) 
and $x^{\vee} $ ($x_j^{\vee}=\lambda_{N-j+1}^{\vee}+j$) for the conditions $M=8$,
$N=4$ and $\lambda=(4,3,1,1)$.  One sees that
the positions of the particles corresponding to $\lambda^{\vee}$ are 
related to those corresponding to $\lambda$ after a $180^{\circ}$ rotation.
}
\label{config}
\end{figure}

The graphical description of the wavefunction 
$\bra x_1\cdots x_N|\psi(\{u\}_N)\ket$ \eqref{wavefunctionone}
is also depicted in Figure~\ref{wave}.   Due to the invariance
of the Boltzmann weights under a $180^{\circ}$ rotation and $[B(u),B(v)]=0$
($[C(u),C(v)]=0$),  the graphical description
of the wavefunction is also invariant under the rotation. One easily
finds that the rotated graph corresponds to the dual wavefunction
$\bra \psi(\{u_N\})|x^{\vee}_1\cdots x^{\vee}_N\ket$ where
 the positions of the  particles $x^{\vee}_j$ is given by \eqref{reverse}. 
After transforming
$x^{\vee}\to x$ which corresponds to the transformation 
$\lambda\to\lambda^{\vee}$, one finds \eqref{wavefunctiontwo}
is valid if \eqref{wavefunctionone} holds.

\begin{figure}[ttt]
\begin{center}
\includegraphics[width=0.7\textwidth]{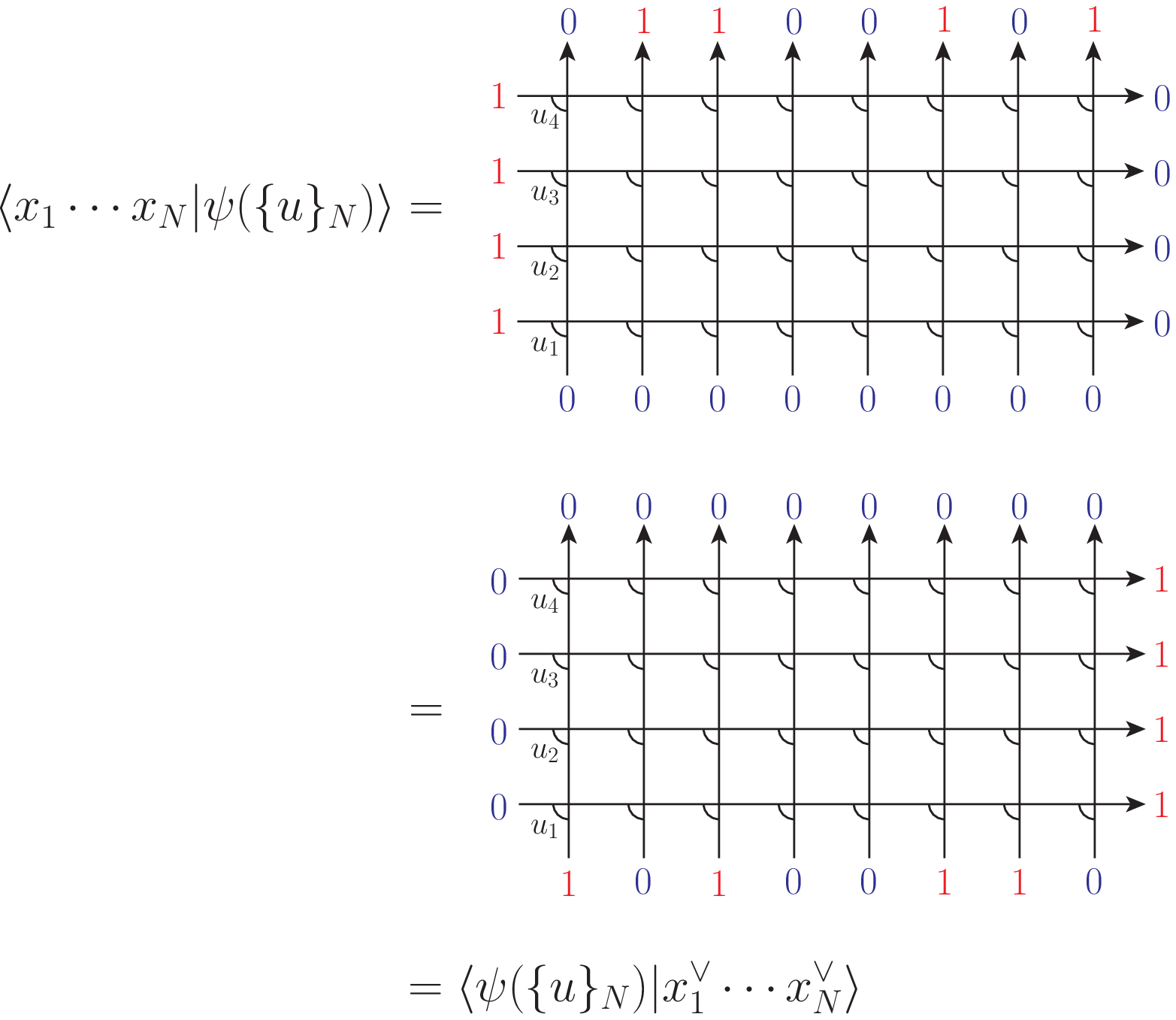}
\end{center}
\caption{A graphical description of the wavefunction \eqref{wavefunctionone}
 for $M=8$, $N=4$ and $\lambda=(4,3,1,1)$ 
(same conditions in Figure~\ref{config}). To go from the first equality to 
the second equality, we use the invariance of the Boltzmann weights
under a $180^{\circ}$ rotation and the commutativity of the $B$- and 
$C$-operators $[B(u),B(v)]=0$ and $[C(u),C(v)]=0$.
}
\label{wave}
\end{figure}

One can show the following Cauchy identity holds for the
Grothendieck polynomials, which is obtained by comparing
the determinant representations for the scalar product 
$\bra\psi( \{v\}_N)|\psi(\{u\}_N)\ket$ \cite{MS} 
to that obtained by multiplying \eqref{wavefunctionone} by 
\eqref{wavefunctiontwo} and then by summing over all 
possible configurations $1\le x_1<\cdots<x_N\le M$.

\begin{theorem}{\rm \cite{MS}}
The following Cauchy identity
for the Grothendieck polynomials holds.
\begin{align}
&\sum_{\lambda \subseteq L^N} G_\lambda(z_1,\dots,z_N;\beta)
G_{\lambda^\vee}(w_1,\dots,w_N;\beta)
\nonumber \\
=&\prod_{1 \le j<k \le N} \frac{1}{(z_j-z_k)(w_k-w_j)}
\mathrm{det}_N
\left[
\frac{z_j^{L+N}(1+\beta w_k)^{N-1}-w_k^{L+N}(1+\beta z_j)^{N-1}}
{z_j-w_k}
\right], \label{cauchy}
\end{align}
where the Young diagram
$\lambda^\vee=(\lambda_1^\vee,\dots,\lambda_N^\vee)$
is given by the Young diagram $\lambda=(\lambda_1,\dots,\lambda_N)$
as $\lambda_j^\vee=L-\lambda_{N+1-j}$.
\end{theorem}
\noindent
Here we have set $L=M-N$, but the above formula holds for any
$L \ge 0$.
As a limiting case of the Cauchy identity, we have also
derived the summation formula for the Grothendieck polynomials.
\begin{theorem} {\rm \cite{MS}}
The following summation for the Grothendieck polynomials holds.
\begin{align}
&\sum_{\lambda\subseteq L^N}(-\beta)^{\sum_{j=1}^N\lambda_j}
G_{\lambda}(z_1,\dots,z_N;\beta)=\prod_{1\le j<k\le N}\frac{1}{z_k-z_j}
\mathrm{det}_N V
\label{sum-Gr}
\end{align}
with an $N\times N$ matrix $V$ whose matrix elements are
\begin{align}
&
V_{jk}=\sum_{m=0}^{j-1}(-1)^m (-\beta)^{j-N}
\binom{L+N}{m}(1+\beta z_k)^{m-j+N-1}
\quad (1\le j \le N-1), \nn \\
&
V_{Nk}=-\sum_{m=\mathrm{max}(N-1,1)}^{L+N} (-1)^m \binom{L+N}{m}
(1+\beta z_k)^{m-1}.
\end{align}
\end{theorem}

 Let us comment on a possible extension of the method to another 
integrable model also defined by \eqref{RLL}. 
In fact,  for given $R$-matrix 
\eqref{Rmatrix}, the $L$-operator \eqref{loperator}
is  not the unique solution to \eqref{RLL}. Indeed,
the non-Hermitian phase model discussed in section~\ref{NHPM} is 
constructed by another solution to \eqref{RLL}.
As mentioned before, the quantum space on which the $L$-operator 
\eqref{loperator} acts is the tensor product of copies of two-dimensional 
space spanned by the empty state $|0\ket$ and particle-occupied 
states $|1\ket$, i.e.,  the double occupancy is  forbidden. 
In this sense, the corresponding quantum system \eqref{Baxter} 
is interpreted as a fermion system. (More precisely, there exists 
a one-to-one correspondence between the spin system \eqref{Baxter} 
and a fermion system through the Jordan-Wigner transformation.) 
On the other hand, the phase model (see \eqref{phase} and \eqref{Lop-boson}
in section~\ref{NHPM})
is a boson system: the quantum space is defined as the tensor 
product of bosonic Fock spaces whose dimension is infinite. At first glance
it seems there is little connection between the fermion system \eqref{Baxter}
and the bosonic system \eqref{phase}, but by definition the algebraic
relations of the both $B$-operators (or $C$-operators) constructing 
the $N$-particle states are  completely the same. Moreover, as 
shown later, the $N$-particle states for the phase model can be uniquely
mapped to those for a fermion model \eqref{Baxter}, and vice versa. These 
observations intuitively indicate
that there is a close correspondence between the fermion \eqref{Baxter} 
and the boson \eqref{phase} models.  This intuition is true.  Indeed, the 
wavefunctions for the both models can be given by the Grothendieck polynomials. 
To show this, first we give an addition theorem satisfied by the  
Grothendieck polynomials, introducing the skew Grothendieck polynomials.

\section{Skew Grothendieck polynomials and addition theorem}
%
The relations \eqref{wavefunctionone} and \eqref{wavefunctiontwo}
between the wavefunctions and the Grothendieck polynomials
lead us to the natural definition of the single variable skew Grothendieck polynomials.

\begin{definition} (cf. \cite{Mc})
The single variable skew Grothendieck polynomial
is defined in terms of the $B$-operator of the five-vertex model:
\begin{align}
&G_{\mu/\lambda}(z;\beta):=\bra y_1 \cdots y_{N+1}|
(-\beta)^N u^{1-M} B(u)|x_1 \cdots x_N \ket, \label{skewone} 
\end{align}
where
$z=-\beta^{-1}-u^{-2}$, and
$\lambda=(\lambda_1,\dots,\lambda_N)$
($M-N \ge \lambda_1 \ge \cdots \ge \lambda_{N} \ge 0$)
and $\mu=(\mu_1,\dots,\mu_{N+1})$
($M-N-1 \ge \mu_1 \ge \cdots \ge \mu_{N+1} \ge 0$)
are  the Young diagrams related to the particle configurations
$x=(x_1,\dots,x_N)$ ($x_j=\lambda_{N-j+1}+j$) and 
$y=(y_1,\dots,y_{N+1})$ ($y_j=\mu_{N-j+2}+j$), respectively.
\end{definition}
\noindent
We shall see later that this is a natural extension of the
skew Schur polynomials.
\begin{proposition}\label{skew}
The skew Grothendieck polynomial $G_{\mu/\lambda}(z;\beta)$
defined in \eqref{skewone} can be given by in terms of the
$C$-operator:
\begin{align}
G_{\mu/\lambda}(z;\beta)=
\bra x^{\vee}_1 \cdots x^{\vee}_N| (-\beta)^N u^{1-M} C(u)
|y_1^{\vee} \cdots y_{N+1}^{\vee} \ket,
\label{skewtwo}
\end{align}
or equivalently,
\begin{align}
G_{\mu^{\vee}/\lambda^{\vee}}(z;\beta)=
\bra x_1 \cdots x_N| (-\beta)^N u^{1-M} C(u)
|y_1 \cdots y_{N+1} \ket,
\label{skewthree}
\end{align}
where the particle positions $x^{\vee}=(x^{\vee}_1,\dots,x^{\vee}_N)$ and
$y^{\vee}=(y^{\vee}_1,\dots,y^{\vee}_{N+1})$ are, respectively, defined as
$x_j^{\vee}=\lambda_{N-j+1}^{\vee}+j$ and $y_j^{\vee}=
\mu_{N-j+2}^{\vee}+j$. 
\end{proposition}
\begin{proof}
The graphical argument is useful to show \eqref{skewtwo} and \eqref{skewthree}.
As discussed in  the wavefunctions 
(see below Theorem~\ref{th-wave} and Figure~\ref{wave}), 
the definition \eqref{skewone} is invariant under a $180^{\circ}$ rotation.
The rotated graph corresponds to 
$\bra x^{\vee}_1 \cdots x^{\vee}_N| (-\beta)^N u^{1-M} C(u)
|y_1^{\vee} \cdots y_{N+1}^{\vee} \ket$. Thus we have \eqref{skewtwo}.
Transforming the variables as $x_j^{\vee}\to x_j=\lambda_{N-j+1}+j$ ($1\le j\le N$)
and $y_j^{\vee} \to y_j=\mu_{N-j+2}+j$ ($1\le j\le N+1$)
which, respectively,  correspond to the transformations  $\lambda\to\lambda^{\vee}$
and $\mu\to\mu^{\vee}$,  one obtains  \eqref{skewthree}.
\end{proof}

Let us define the ordering on the Young diagrams for later purpose.
\begin{definition}
For two Young diagrams $\mu=(\mu_1,\mu_2,\dots,\mu_{N+1})$ and
$\lambda=(\lambda_1,\lambda_2,\dots,\lambda_N)$,
we say that $\mu$ and $\lambda$ interlace,
if and only if $\mu_j \ge \lambda_j \ge \mu_{j+1}$ $(j=1,\dots,N)$,
and write this relation as $\mu \succ \lambda$. 
Correspondingly, 
we write $y\succ x$ for the particle configurations $y=(y_1,\dots,y_{N+1})$
($y_j=\mu_{N-j+2}+j$)
and  $x=(x_1,\dots,x_N)$ ($x_j=\lambda_{N-j+1}+j$), if and only if
$\mu \succ \lambda$ or equivalently $y_{j}\le x_j< y_{j+1}$ $(j=1,\dots,N)$ holds.
\end{definition}
\noindent
In Figure~\ref{Interlacing-1} (resp. Figure~\ref{Interlacing-2}), an example of  the
 interlacing
(resp. non-interlacing) partitions and the corresponding particle configurations 
are depicted. It immediately follows that
\begin{align}
\mu \succ \lambda \Longleftrightarrow \mu^{\vee}\succ \lambda^{\vee}.
\label{correspond}
\end{align}

\begin{figure}[ttt]
\begin{center}
\includegraphics[width=0.9\textwidth]{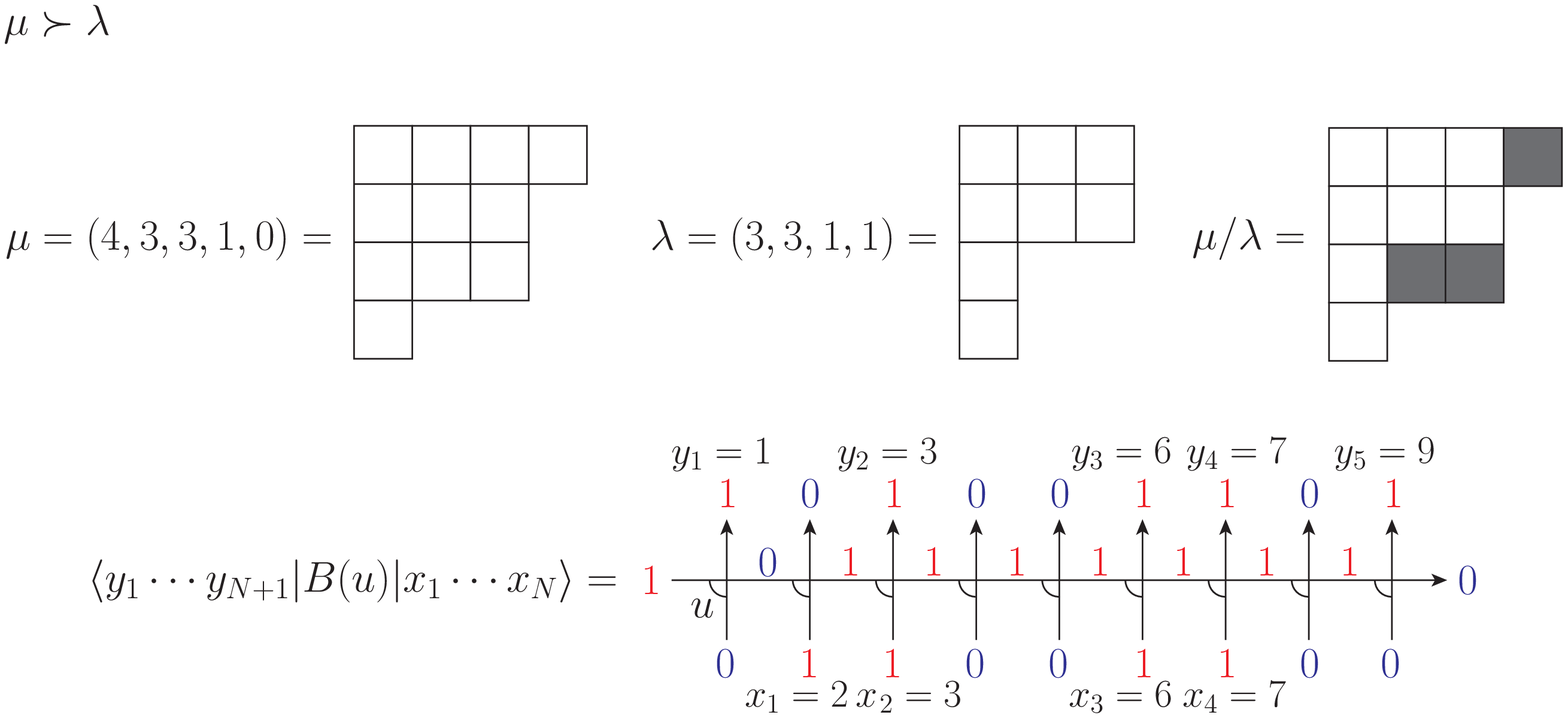}
\end{center}
\caption{An example of the interlacing partition functions 
$\mu \succ \lambda$. Here we have set
$\mu=(4,3,3,1,0)$ and $\lambda=(3,3,1,1)$. The skew Young diagram $\mu/\lambda$ is
 depicted as the gray boxes. The 
input (resp. output) state denotes the particle configuration corresponding to $\lambda$ 
(resp. $\mu$). For interlacing partitions $\mu\succ\lambda$,  
the matrix element $\bra y_1 \cdots y_{N+1}|B(u) |x_1\cdots x_N\ket$ is non-zero
for the generic value of $u$. }
\label{Interlacing-1}
\end{figure}

\begin{figure}[ttt]
\begin{center}
\includegraphics[width=0.9\textwidth]{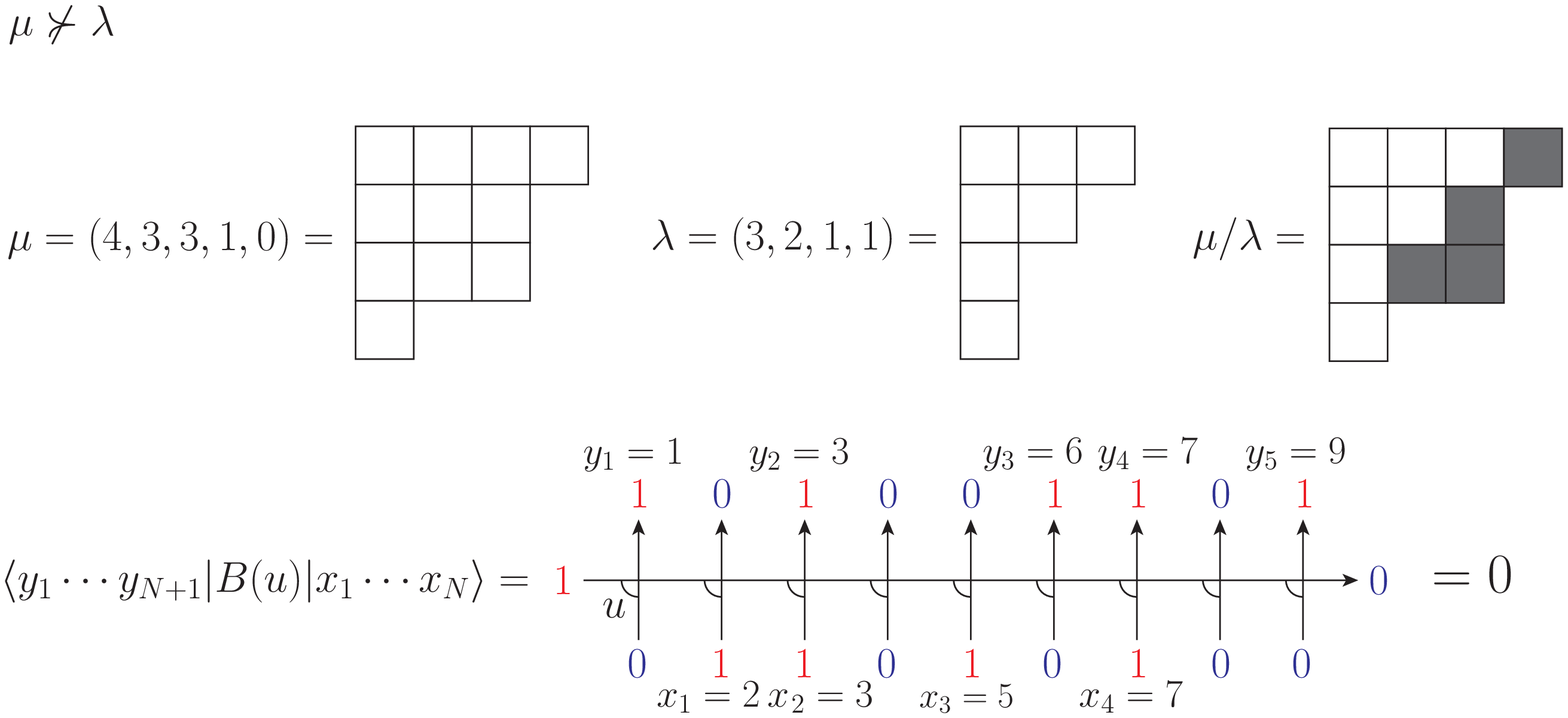}
\end{center}
\caption{
An example of the non-interlacing partition functions $\mu\nsucc \lambda$. 
Here we have set $\mu=(4,3,3,1,0)$ and $\lambda=(3,2,1,1)$.  The 
input (resp. output) state denotes the particle configuration corresponding to $\lambda$ 
(resp. $\mu$). For non-interlacing partitions $\mu\nsucc\lambda$, one sees
$\bra y_1 \cdots y_{N+1}|B(u) |x_1\cdots x_N\ket=0$.}
\label{Interlacing-2}
\end{figure}

The single variable skew Grothendieck polynomials \eqref{skewone}
is given by the following explicit expression.
\begin{proposition}
The single variable skew Grothendieck polynomials $G_{\mu/\lambda}(z;\beta)$
can be explicitly expressed as
\begin{align}
G_{\mu/\lambda}(z;\beta)
=
\begin{cases}
z^{\sum_{j=1}^{N+1} \mu_j-\sum_{j=1}^N \lambda_j}
\prod_{j=1}^N (1+\beta z-\beta z \delta_{\mu_{j+1}\, \lambda_j})
& \mu \succ \lambda \\
0 & \text{otherwise}
\end{cases}.
\label{skewexpression}
\end{align}
The case $\beta=0$ reduces to the single variable skew Schur polynomials:
$G_{\mu/\lambda}(z;0)=s_{\mu/\lambda}(z)
=z^{\sum_{j=1}^{N+1} \mu_j-\sum_{j=1}^N \lambda_j}$.
\end{proposition}
\begin{proof}
We show \eqref{skewexpression} by explicit evaluation of
the definition \eqref{skewone}.
From the graphical description (see Figure~\ref{Interlacing-2}, for instance), 
we find $\bra y_1\cdots y_{N+1}| B(u) |x_1\cdots x_N \ket=0$
for $\mu\nsucc \lambda$. Thus $G_{\mu/\lambda}(z;\beta)=0$ holds for $\mu\nsucc \lambda$.
The first equality in \eqref{skewexpression} can be shown by the following decomposition:
\begin{align}
B(u)| x_1\cdots x_N\ket={}_a\bra 0 |\prod_{j=1}^{N+1}\prod_{l=x_{j-1}+1}^{x_j} L_{al}(u)
|1 \ket_a \otimes 
 \left\{\otimes_{k=x_{j-1}+1}^{x_j-1} | 0 \ket_k \right\}
\otimes | 1 \ket_{x_j} ,
\label{decomp}
\end{align}
where $x_0=0$ and $x_{N+1}=M$.
Using the following relations,
\begin{align}
&\prod_{l=x_{j-1}+1}^{x_j}L_{al}(u)|1 \ket_a \otimes 
 \left\{\otimes_{k=x_{j-1}+1}^{x_j-1} | 0 \ket_k \right\}
\otimes | 1 \ket_{x_j} \nonumber \\
&\qquad=\sum_{x_{j-1}+1 \le y_j \le x_j}
(-\beta^{-1}u-u^{-1})^{y_j-x_{j-1}-1}
\left\{ u^{x_j-y_j-1}(1-\delta_{x_j\, y_j})
-\beta^{-1} u \delta_{x_j\, y_j} \right\} \nonumber \\
&\qquad \quad \times | 1 \ket_a \otimes 
\left\{ \otimes_{k=x_{j-1}+1}^{y_j-1}|0 \ket_k \right\} \otimes
|1 \ket_{y_j} \otimes
\left\{ \otimes_{k=y_j+1}^{x_j}|0 \ket_k \right\}   \qquad (1 \le j \le N),
\label{parts}
\\
&{}_a\bra 0| \prod_{l=x_N+1}^M L_{al}(u)|1 \ket_a \otimes
\left\{ \otimes_{k=x_N+1}^M |0 \ket_k \right\} \nonumber \\
&\qquad=\sum_{x_N+1 \le y_{N+1} \le M} (-\beta^{-1}u-u^{-1})^{y_{N+1}-x_N-1}
u^{M-y_{N+1}} \nonumber \\
&\qquad \quad \times \left\{ \otimes_{k=x_N+1}^{y_{N+1}-1} |0 \ket_k \right\}
\otimes
|1 \ket_{y_{N+1}} \otimes \left\{ \otimes_{k=y_{N+1}+1}^M | 0 \ket_k \right\},
\label{parts-two}
\end{align}
which are directly obtained by the graphical representation shown in 
Figure~\ref{skew-fig},
we find that \eqref{decomp} yields
\begin{align}
B(u)|x_1 \cdots x_N \ket
=&\sum_{y\succ x}
(-\beta^{-1}u-u^{-1})^{-1-N+\sum_{j=1}^{N+1}y_j-\sum_{j=1}^N x_j}
u^{M-y_{N+1}} \nonumber \\
&\times\prod_{j=1}^N \{ u^{x_j-y_j-1}(1-\delta_{x_j \, y_j})
-\beta^{-1} u \delta_{x_j \, y_j} \} | y_1 \cdots y_{N+1} \ket.
\label{actionone}
\end{align}

\begin{figure}[ttt]
\begin{center}
\includegraphics[width=0.9\textwidth]{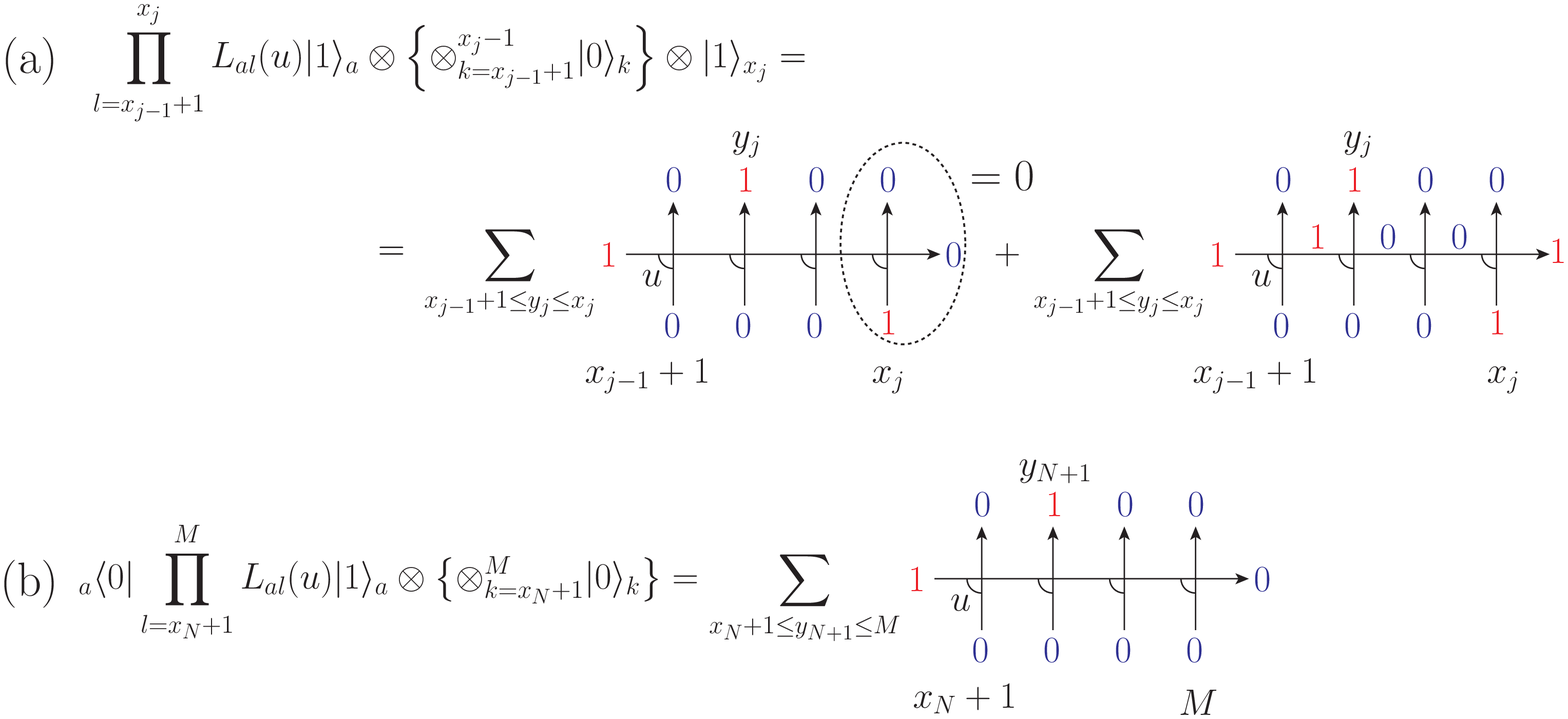}
\end{center}
\caption{(a): The graphical description of \eqref{parts}.  The first term of the
right hand vanishes because the Boltzmann weight surrounded by the broken 
line is equal to zero.  The insertion of 
the weights shown in Figure~\ref{weight} into the second term yields \eqref{parts}.
(b): The graphical description of \eqref{parts-two}.}
\label{skew-fig}
\end{figure}
Translating \eqref{actionone} into  the language of Young diagrams
by $x_j=\lambda_{N-j+1}+j$ and $y_j=\lambda_{N-j+2}+j$, and
using $y\succ x \Longleftrightarrow\mu\succ \lambda$, we have the
first equality in \eqref{skewexpression}.
\end{proof}
Combining the relation between the wavefunction
and the Grothendieck polynomial
\eqref{wavefunctionone}
and the definition of the skew Grothendieck polynomials
\eqref{skewone}, we have the following addition theorem.
\begin{theorem}\label{addition}
The following relation between
the Grothendieck and skew Grothendieck polynomials holds
\begin{align}
G_\mu(z_1,\dots,z_N,z_{N+1};\beta)=\sum_{\mu \succ \lambda}
G_{\mu/\lambda}(z_{N+1};\beta)G_\lambda(z_1,\dots,z_N;\beta), \label{relation}
\end{align}
which recovers the one for the Schur and skew Schur polynomials at $\beta=0$.
\end{theorem}
\begin{proof}
This follows from evaluating the $(N+1)$-particle state  by using \eqref{wavefunctionone}
and \eqref{skewone} as
\begin{align}
(-\beta)^{(N+1)N/2} &\prod_{j=1}^{N+1} u_j^{1-M} B(u_j)|\Omega \ket=
(-\beta)^N u_{N+1}^{1-M} B(u_{N+1})
(-\beta)^{N(N-1)/2}
\prod_{j=1}^{N} u_j^{1-M} B(u_j)|\Omega \ket \nonumber \\
=&(-\beta)^N u_{N+1}^{1-M} B(u_{N+1})
\sum_{\lambda}G_\lambda(z_1,\dots,z_N;\beta)
|x_1(\lambda) \cdots x_N(\lambda) \ket \nonumber \\
=&\sum_{\mu \succ \lambda}G_{\mu/\lambda}(z_{N+1};\beta)
G_\lambda(z_1,\dots,z_N;\beta)|y_1(\mu) \cdots y_{N+1}(\mu) \ket,
\label{expansiontwotimes}
\end{align}
and comparing with
\begin{align}
(-\beta)^{(N+1)N/2} \prod_{j=1}^{N+1} u_j^{1-M} B(u_j)|\Omega \ket
=\sum_{\mu}G_\mu(z_1,\dots,z_{N+1};\beta)
|y_1(\mu) \cdots y_{N+1}(\mu) \ket.
\end{align}
Note that $\lambda$ in the summation \eqref{expansiontwotimes}
can be restricted from $\lambda \subseteq (M-N)^N$
to $\lambda \subseteq (M-N-1)^N$
since $\mu \subseteq (M-N-1)^{N+1}$ and $G_{\mu/\lambda}(z;\beta)=0$
unless $\mu \succ \lambda$.
\end{proof}
The relation \eqref{relation}
is the consequence of the action of a $B$-operator
on the wavefunction of the $N$-particle state,
from which also justifies the definition \eqref{skewone}
of the skew Grothendieck polynomials.
In the next section, we use this addition theorem
to show that the wavefunction of the non-Hermitian phase model 
can also be expressed as Grothendieck polynomials.

The repeated application of the addition theorem \eqref{relation}
leads to the following corollary.
\begin{corollary}\label{decomposition}
The Grothendieck polynomials can be expressed in terms of
the single variable skew Grothendieck polynomials as
\begin{align}
G_{\lambda}(z_1,\dots,z_N;\beta)
&=\sum_{\lambda=\lambda^{(0)} \succ \lambda^{(1)} \succ \cdots \succ \lambda^{(N)}
=\emptyset}
\prod_{j=1}^N G_{\lambda^{(j-1)}/\lambda^{(j)}}(z_j;\beta).
\label{decompositionskew}
\end{align}
\end{corollary}

Before closing this section, we define the multivariable skew 
Grothendieck polynomials for completeness of the paper.
The multivariable skew Grothendieck is naturally defined by
multiplying the single variable skew Grothendieck polynomials.
\begin{definition}
The multivariable skew Grothendieck polynomials is defined as
\begin{align}
G_{\lambda/\nu}(z_1,\dots,z_n;\beta)&:=
\sum_{\lambda^{(1)}\succ \cdots \succ \lambda^{(n-1)}}
\prod_{j=1}^n
G_{\lambda^{(j-1)}/\lambda^{(j)}}(z_j;\beta)
\end{align}
where $\lambda=\lambda^{(0)}$ and $\nu=\lambda^{(n)}$.
\end{definition}
The combination of Corollary~\ref{decomposition} and Theorem~\ref{addition}
leads to the following addition theorem.
\begin{theorem}
The following relation between the Grothendieck polynomials and
the (multivariable) skew Grothendieck polynomials holds:
\begin{align}
G_{\lambda}(z_1,\dots,z_n, w_{1},\dots, w_m;\beta)=\sum_{\nu} 
G_{\lambda/\nu}(z_1,\dots,z_n;\beta)G_{\nu}(w_1,\dots,w_m;\beta),
\end{align} 
which recovers the one for the Schur and skew Schur polynomials at $\beta=0$.
\end{theorem}

\section{Non-Hermitian phase model}\label{NHPM}
%
In this section, we introduce the non-Hermitian phase model \cite{BN},
which can be solved by the algebraic Bethe ansatz. The phase model is
a boson system characterized by the generators 
$\phi$, $\phi^\dagger$, $N$ and $\pi$ acting on a bosonic 
Fock space $\mathcal{F}$ spanned by orthonormal basis  
$| n \ket \ (n=0,1,\dots, \infty)$. Here the number $n$ 
indicates the occupation number of bosons.  The generators 
$\phi$, $\phi^\dagger$, $N$ and $\pi$ are, respectively, 
the annihilation, creation, number and vacuum projection operators,
whose actions on $\mathcal{F}$ are, respectively, defined as
\begin{align}
\phi|0 \ket=0,  \quad 
\phi|n \ket=|n-1 \ket, \quad
\phi^\dagger|n \ket=|n+1 \ket, \quad 
N|n \ket=n|n \ket, \quad
\pi|n \ket=\delta_{n\,0}|n \ket.
\end{align}
Thus the operator forms are explicitly given by
\begin{align}
\phi=\sum_{n=0}^\infty |n \ket \bra n+1|, \quad
\phi^\dagger=\sum_{n=0}^\infty |n+1 \ket \bra n|, \quad
N=\sum_{n=0}^\infty n |n \ket \bra n|, \quad
\pi=|0 \ket \bra 0|.
\end{align}
These operators generate an algebra referred to as the phase algebra:
\begin{align}
[\phi, \phi^\dagger]=\pi, \quad
[N, \phi]=-\phi, \quad
[N, \phi^\dagger]=\phi^\dagger.
\end{align}
The non-Hermitian phase model \cite{BN,SW} under the periodic boundary 
condition is defined by the
following Hamiltonian:
\begin{align}
\mathcal{H}=\sum_{j=0}^{M-1} (\phi_{j+1}^\dagger \phi_j-\beta \pi_j).
\label{phase}
\end{align}
The Hamiltonian acts on the tensor product of Fock spaces
$\otimes_{j=0}^{M-1} \mathcal{F}_j$, whose basis is given by
$|\{ n \}_M \ket :=
\otimes_{j=0}^{M-1}|n_j \ket_j$, $n_j=0,1,\dots,\infty$.
We denote a dual state of $|\{ n \}_M \ket$ as
$\bra \{ n \}_M| := \otimes_{j=0}^{M-1} {}_j \bra n_j|$.
The operators $\phi_j$, $\phi_j^\dagger$, $N_j$ and $\pi_j$
act on the Fock space $\mathcal{F}_j$ as $\phi$, 
$\phi^\dagger$, $N$ and $\pi$,
and the other Fock spaces $\mathcal{F}_k, k \neq j$ as
an identity. The term including $\beta$ in \eqref{phase}
denotes an on-site interaction: $\beta>0$, $\beta=0$ and $\beta<0$
correspond to repulsive, free and attractive interactions, respectively.

The Hamiltonian is quantum integrable,
and a special point $\beta=-1$ describes a stochastic process
without exclusion
called the totally asymmetric zero range process (TAZRP), i.e., a stochastic process for a system of
bosons so that each site can be occupied by arbitrary number of particles,
which is in contrast to the TASEP where each site
can be occupied by at most one particle.

We can make an analysis on the non-Hermitian phase model
by the quantum inverse scattering method.
The basic object is the following $L$ operator
\begin{align}
\mathcal{L}_{a j}(v)=
\begin{pmatrix}
v^{-1}-\beta v \pi_j & \phi_j^\dagger  \\
\phi_j & v
\end{pmatrix},
\label{Lop-boson}
\end{align}
acting on the tensor product $W_a \otimes \mathcal{F}_j$ of the
complex two-dimensional space $W_a$ and the Fock space at the
$j$th site $\mathcal{F}_j$. See also Figure~\ref{boson-L} for
a pictorial representation of the $L$-operator \eqref{Lop-boson},
which allows for an intuitive understanding of the subsequent calculations.
The $L$-operator satisfies the intertwining relation ($RLL$-relation)
\begin{align}
R_{ab}(u,v)\mathcal{L}_{a j}(u)\mathcal{L}_{b j}(v)=
\mathcal{L}_{b j}(v)\mathcal{L}_{a j}(u)R_{ab}(u,v),
\label{RLL-boson}
\end{align}
which acts on $W_a \otimes W_b \otimes \mathcal{F}_j$.
The $R$ matrix $R(u,v)$ is the same as the one for the integrable
five-vertex model \eqref{Rmatrix}. The auxiliary space $W_a$ is the
complex two-dimensional space, which is the same as that for the
integrable five-vertex model, while the quantum space $\mathcal{F}_j$
is the infinite-dimensional bosonic Fock space.

\begin{figure}[tt]
\begin{center}
\includegraphics[width=0.7\textwidth]{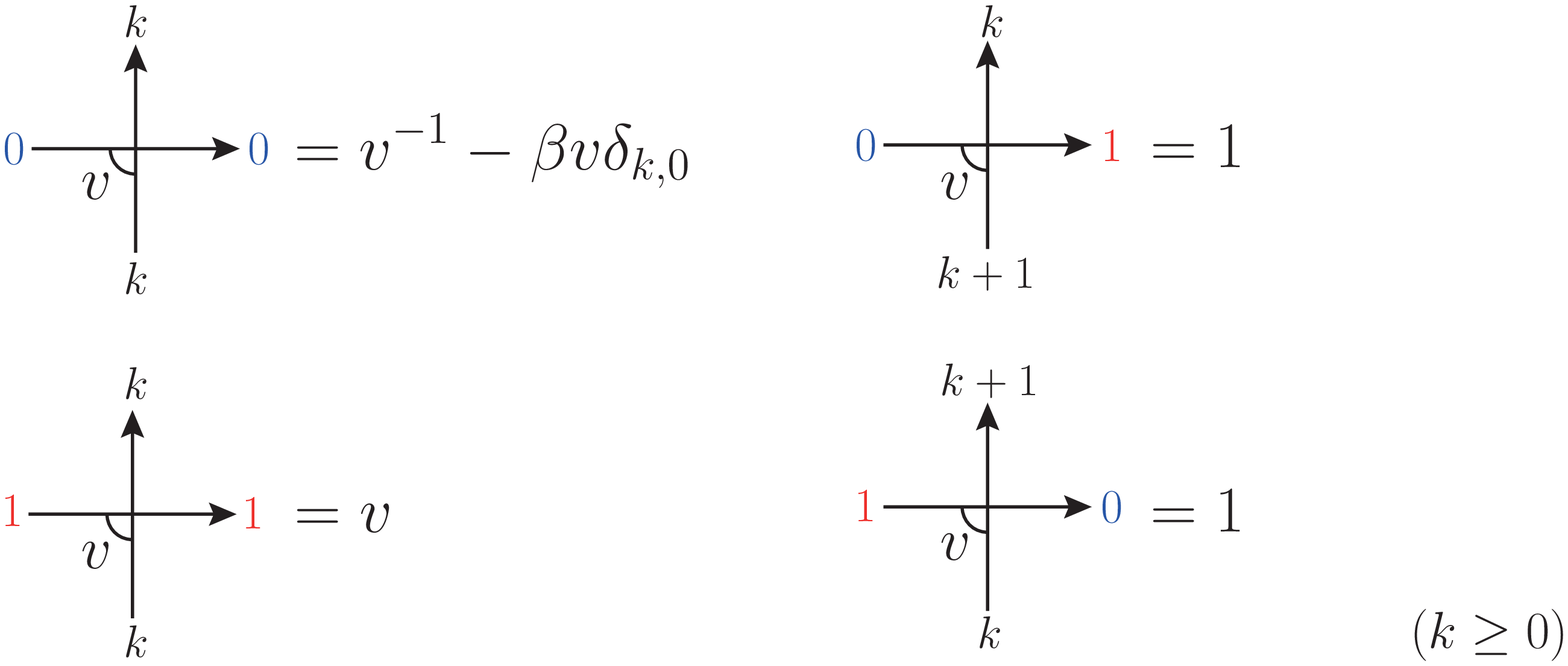}
\end{center}
\caption{The non-zero elements of the $L$-operator \eqref{Lop-boson}
for the non-Hermitian phase model. The variables on the left-arrows
take $0$ and $1$, since  the auxiliary space for
the phase model is two-dimensional space which is 
the same as the one for the five-vertex model.
On the other hand, the variables on the up-arrows take
infinite values $0,1,\dots,\infty$ which reflects the fact that
the quantum space of the phase model is infinite dimensions.
Note that the weights are invariant under a $180^{\circ}$ rotation.
}
\label{boson-L}
\end{figure}

From the $L$-operator, we construct the monodromy matrix
\begin{align}
\mathcal{T}_{a}(v)=\mathcal{L}_{a M-1}(v) \cdots \mathcal{L}_{a 0}(v)
=
\begin{pmatrix}
\mathcal{A}(v) & \mathcal{B}(v)  \\
\mathcal{C}(v) & \mathcal{D}(v)
\end{pmatrix}_{a}, 
\label{monodromy}
\end{align}
which acts on $W_a \otimes (\mathcal{F}_0\otimes\dots\otimes 
\mathcal{F}_{M-1})$. Tracing
out the auxiliary space, one defines the transfer matrix 
$\tau(u)\in \mathrm{End} (\mathcal{F}^{\otimes M})$:
\begin{align}
\tau (v)=\Tr_{W_a} \mathcal{T}_{a}(v). 
\label{TM}
\end{align}
The repeated applications of the $RLL$-relation leads to the
intertwining relation
\begin{align}
R_{ab}(u,v)\mathcal{T}_{a}(u)\mathcal{T}_{b}(v)=
\mathcal{T}_{b}(v)\mathcal{T}_{a}(u)R_{a b}(u,v) \label{RTT}.
\end{align}
Some elements of the relation \eqref{RTT} are
\begin{align}
&\mathcal{C}(u)\mathcal{B}(v)=g(u,v) 
\left[ \mathcal{A}(u)\mathcal{D}(v)-\mathcal{A}(v)\mathcal{D}(u) \right],
\nn  \\
&\mathcal{A}(u)\mathcal{B}(v)=f(u,v) \mathcal{B}(v) \mathcal{A}(u)+
g(v,u)\mathcal{B}(u)\mathcal{A}(v), \nn \\
&\mathcal{D}(u)\mathcal{B}(v)=f(v,u) \mathcal{B}(v) \mathcal{D}(u)+
g(u,v)\mathcal{B}(u)\mathcal{D}(v), \nn \\
&\left[\mathcal{B}(u),\mathcal{B}(v)\right]
=\left[\mathcal{C}(u),\mathcal{C}(v)\right]=0.
\label{commutation}
\end{align}
The above relations are  completely the same as those satisfied by $A(u)$, $B(u)$,
$C(u)$ and $D(u)$  for the integrable five-vertex model, since the $RLL$-relation
\eqref{RLL-boson} is the same as \eqref{RLL}.
Thanks to the $RTT$-relation \eqref{RTT},
the transfer matrix $\tau(u)$ mutually 
commutes, i.e.,
\begin{align}
[\tau(u),\tau(v)]=0.
\label{commutative}
\end{align}
The Hamiltonian
can be obtained by the derivative of  the transfer matrix 
with respect to the spectral parameter:
\begin{align}
\mathcal{H}
=
\left.\frac{\partial}{\partial v^2} (v^M \tau(v))
\right|_{v=0}.
\label{Baxter-B}
\end{align}

The arbitrary $N$-particle state $|\Psi(\{v \}_N) \ket$
(resp. its dual $\bra \Psi(\{v \}_N)|$) 
(not normalized) with $N$ spectral parameters
$\{ v \}_N=\{ v_1,\dots,v_N \}$
is constructed by a multiple action
of $\mathcal{B}$ (resp. $\mathcal{C}$) operator on the vacuum state 
$|\Omega \ket:= | 0^{M} \ket:=|0\ket_0\
\otimes \dots \otimes |0\ket_{M-1}$
(resp. $\bra \Omega| := \bra 0^{M}|:=
{}_0\bra 0|\otimes\dots \otimes{}_{M-1}\bra 0|$):
\begin{align}
|\Psi(\{v \}_N) \ket=\prod_{j=1}^N \mathcal{B}(v_j)| \Omega \ket,
\quad
\bra \Psi(\{v \}_N)|=\bra \Omega| \prod_{j=1}^N
\mathcal{C}(v_j).
\label{statevector-B}
\end{align}
The $N$-particle state $|\Psi(\{v \}_N) \ket$
and its dual $\bra \Psi(\{v \}_N)|$
become an eigenvector of the transfer matrix
with the eigenvalue
\begin{align}
\tau(u)=(v^{-1}-\beta v)^M \prod_{k=1}^N \frac{v_k^2}{v_k^2-v^2}
+v^M \prod_{k=1}^N \frac{v^2}{v^2-v_k^2},
\end{align}
if the spectral parameters $\{ v \}_N$
satisfy the Bethe ansatz equation
\begin{align}
(v_j^{-2}-\beta)^M
=(-1)^{N-1} \prod_{k=1}^N \frac{v_j^2}{v_k^2}.
\end{align}
The eigenvalue of the Hamiltonian is given by
\begin{align}
E=-\beta M+\sum_{j=1}^N v_j^{-2}.
\end{align}

\section{Wavefunctions and scalar products}
%
Here and in what follows, we consider the arbitrary off-shell
state, i.e., the parameters $\{v\}_N$ in the $N$-particle state
\eqref{statevector-B} are arbitrary.
The orthonormal basis of
the $N$-particle state $|\Psi(\{v \}_N) \ket$
and its dual $\bra \Psi(\{v \}_N)|$
is given by $| \{ n \}_{M,N} \ket := |n_0\ket_0\
\otimes \dots \otimes |n_{M-1}\ket_{M-1}$
and $\bra \{ n \}_{M,N}| := _{0}\bra n_0| \otimes \dots
\otimes _{M-1} \bra n_{M-1}|$, where $n_0+n_1+\cdots+n_{M-1}=N$.
The wavefunctions can be expanded in this basis as
\begin{align}
&|\Psi(\{v \}_N) \ket
=\sum_{\substack{0 \le n_0,\dots,n_{M-1} \le N \\
n_0+\cdots+n_{M-1}=N}}
\bra \{ n \}_{M,N}|\psi(\{v \}_N) \ket
|\{ n \}_{M,N} \ket, \\
&\bra \Psi(\{v \}_N)|
=\sum_{\substack{0 \le n_0, \dots,n_{M-1} \le N \\
n_0+\cdots+n_{M-1}=N}}
\bra \{ n \}_{M,N}| \bra \psi(\{v \}_N)| \{ n \}_{M,N} \ket.
\end{align}
There is a one-to-one correspondence between
the set $\{ n \}_{M,N}=\{n_0,n_1,\dots,n_{M-1} \}$ ($n_0+n_1+\cdots+n_{M-1}=N$)
and the Young diagram $\lambda=(\lambda_1,\lambda_2,\dots,\lambda_N)$
($M-1 \ge \lambda_1 \ge \lambda_2 \ge \cdots \ge \lambda_N \ge 0$).
Namely, each Young diagram $\lambda$ under the constraint
$\ell(\lambda) \le N$, $\lambda_1 \le M-1$
can be labeled by a set of integers $\{ n \}_{M,N}$ as
$\lambda=((M-1)^{n_{M-1}}, \dots, 1^{n_1},0^{n_0})$.
In Figure~\ref{config-B}, we denote the particle positions
$\{n\}_{M,N}$ and $\{n^{\vee} \}_{M,N}$, which correspond
to the Young diagram $\lambda$ and $\lambda^{\vee}$, respectively. 
From this, one can intuitively 
find that the positions of the particles corresponding to $\lambda^{\vee}$ are 
related to those corresponding to $\lambda$ after a $180^{\circ}$ rotation.
\begin{figure}[ttt]
\begin{center}
\includegraphics[width=0.85\textwidth]{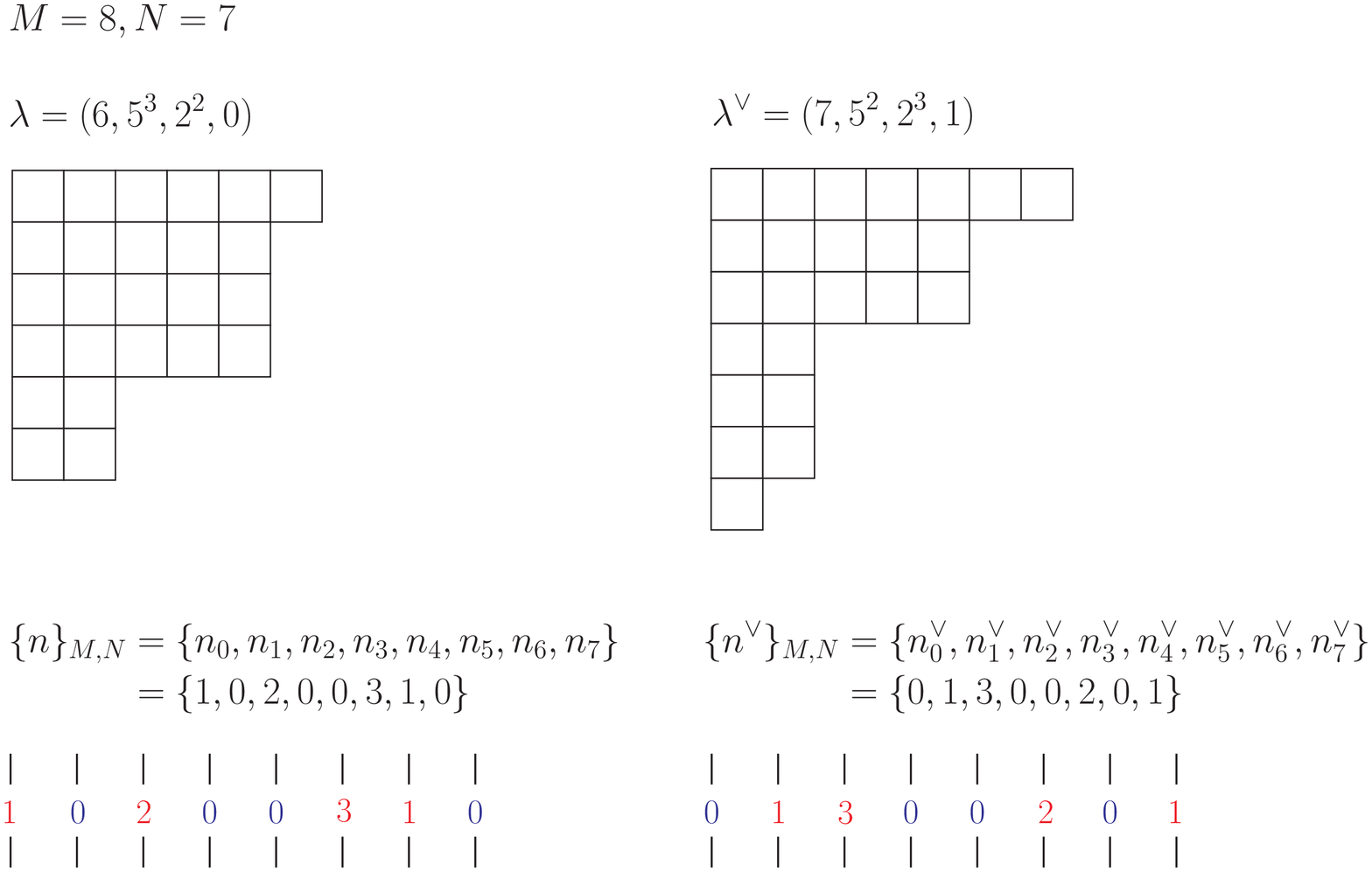}
\end{center}
\caption{An example of the partitions $\lambda$ and $\lambda^{\vee}$
and corresponding particle configurations $\{n\}_{M,N}=\{n_0,\dots n_{M-1}\}$ 
and$\{n^{\vee}\}_{M,N}=\{n^{\vee}_0,\dots n^{\vee}_{M-1}\}$  for 
the conditions $M=8$, $N=7$ and $\lambda=(6,5^3,2^2,0)$.  One sees that
the positions of the particles corresponding to $\lambda^{\vee}$ are 
related to those corresponding to $\lambda$ after a $180^{\circ}$ rotation.
}
\label{config-B}
\end{figure}

The following definition \cite{Wh}
on the ordering on the basis of particle configurations
is useful for later purpose.
\begin{definition} \cite{Wh}
For two configurations
$\{ m \}_{M,N+1}=\{m_0,m_1,\dots,m_{M-1} \}$ $(m_0+m_1+\cdots+m_{M-1}=N+1)$
and
$\{ n \}_{M,N}=\{n_0,n_1,\dots,n_{M-1} \}$ $(n_0+n_1+\cdots+n_{M-1}=N)$,
let $\sum_j^m=\sum_{k=j}^{M-1}m_k$ and $\sum_j^n=\sum_{k=j}^{M-1}n_k$.
We say that the particle configurations
$\{ m \}_{M,N+1}$  and $\{ n \}_{M,N}$ are admissible, if and only if
$0 \le (\sum_j^m-\sum_j^n) \le 1$ ($1 \le j \le M-1$), and write this relation
as $\{ m \}_{M,N+1} \triangleright \{ n \}_{M,N}$.
\end{definition}
\begin{proposition}\label{admissible}
Let $\{m\}_{M,N+1}$ and $\{n\}_{M,N}$ be the particle configurations
described by the Young diagram $\mu=(\mu_1,\dots,\mu_{N+1})$
and $\lambda=(\lambda_1,\dots,\lambda_{N})$. Then 
\begin{align}
\mu\succ\lambda \Longleftrightarrow \{ m \}_{M,N+1} \triangleright \{ n \}_{M,N}.
\end{align}
\end{proposition}
\noindent
In Figure~\ref{Interlacing-B1} (resp. Figure~\ref{Interlacing-B2}), an example of the
 interlacing
(resp. non-interlacing)  partitions and the corresponding admissible (resp. non-admissible) 
particle configurations are depicted.

\begin{figure}[ttt]
\begin{center}
\includegraphics[width=0.75\textwidth]{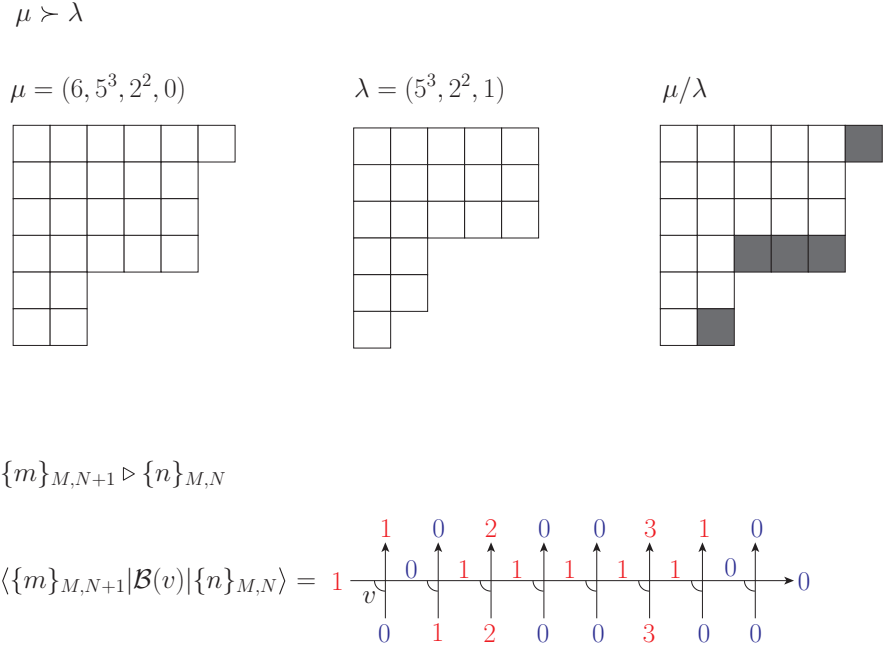}
\end{center}
\caption{An example of the interlacing partition functions 
$\mu \succ \lambda$. Here we have set
$\mu=(6,5^3,2^2,0)$ and $\lambda=(5^3,2^2,1)$. The skew Young diagram 
$\mu/\lambda$ is
 depicted as the gray boxes. The 
input (resp. output) state denotes the particle configuration corresponding to $\lambda$ 
(resp. $\mu$).  The particle configurations are admissible for the 
interlacing partitions. 
For admissible configurations  $\{m\}_{M,N+1} \triangleright\{n\}_{M,N}$,
the matrix element $\bra \{m\}_{M,N+1}|\mathcal{B}(v) |\{n\}_{M,N}\ket$
is non-zero for the generic value of $v$.}
\label{Interlacing-B1}
\end{figure}

\begin{figure}[ttt]
\begin{center}
\includegraphics[width=0.75\textwidth]{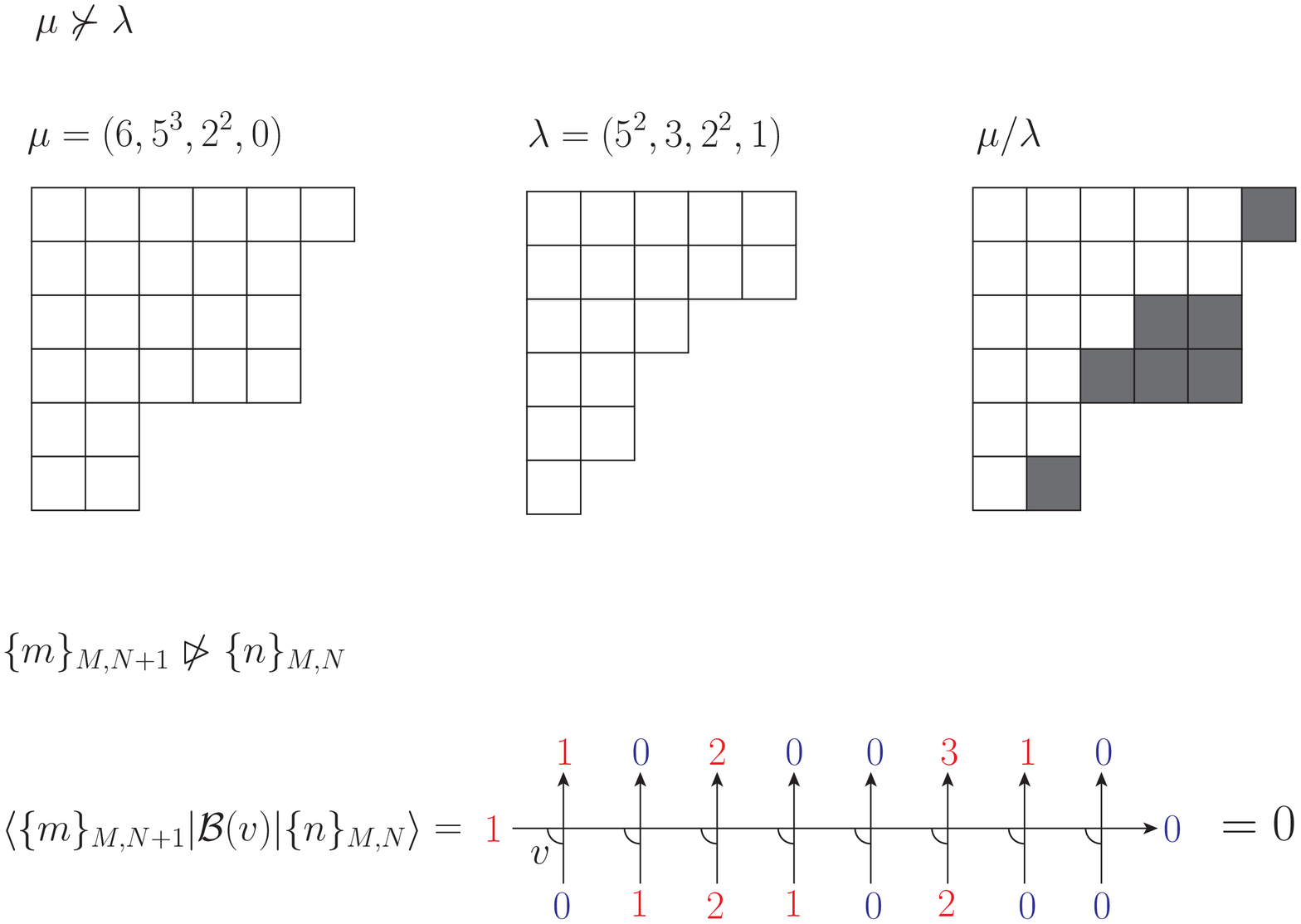}
\end{center}
\caption{
An example of the non interlacing partition functions 
$\mu \nsucc \lambda$. Here we have set
$\mu=(6,5^3,2^2,0)$ and $\lambda=(5^2,3,2^2,1)$.
 For non-interlacing partitions, the particle configurations are not admissible. 
For non-admissible configurations $\{m\}_{M,N+1} \ntriangleright\{n\}_{M,N}$, 
one sees $\bra \{m\}_{M,N+1}|\mathcal{B}(v) |\{n\}_{M,N}\ket=0$.}
\label{Interlacing-B2}
\end{figure}

We show the wavefunctions $\bra \{ n \}_M|\Psi(\{v \}_N) \ket$
and its dual $\bra \Psi(\{v \}_N)| \{ n \}_M \ket$
can be represented in the following determinant forms
which are parametrized by Young diagrams.
\begin{theorem} \label{bosontheorem}
The wavefunctions can be expressed by the Grothendieck polynomials as
\begin{align}
& \bra \{ n \}_{M,N}|\Psi(\{v \}_N) \ket
=\prod_{j=1}^N (v_j^{-1}-\beta v_j)^{M-1}G_\lambda(z_1,\dots,z_N;\beta), 
\label{Grothen-B1} \\
&\bra \Psi(\{v \}_N)| \{ n \}_{M,N} \ket
=\prod_{j=1}^N (v_j^{-1}-\beta v_j)^{M-1}
G_{\lambda^\vee}(z_1,\dots,z_N;\beta) \label{Grothen-B2},
\end{align}
where $z_j^{-1}=v_j^{-2}-\beta$ and $\lambda^\vee=(\lambda_1^\vee,\lambda_2^\vee,\dots,\lambda_N^\vee)$ ($M-1 \ge \lambda_1^\vee \ge \cdots \ge \lambda_N^\vee \ge 0$) is given by the Young diagram $\lambda$ as
$\lambda_j^\vee=M-1-\lambda_{N+1-j}$.
\end{theorem}
\begin{proof}
The second relation \eqref{Grothen-B2} holds if the first equation
\eqref{Grothen-B1} is valid. This follows from an argument similar
to that in the wavefunctions for the five-vertex model. Namely, since
the Boltzmann weights for the phase model 
are invariant under a $180^{\circ}$ (see Figure~\ref{boson-L}) and
the commutativity of the $B$- and $C$-operators \eqref{commutation},
the graphical description of the wave function 
$\bra \{ n \}_{M,N}|\Psi(\{v \}_N) \ket$
is also invariant under 
the rotation (cf. Figure~\ref{wave} for the five-vertex model).
The rotated graph is nothing but the dual wavefunction
$\bra \Psi(\{v \}_N)| \{ n^{\vee} \}_{M,N} \ket$. Transforming
$ \{ n^{\vee} \}\to \{ n \}$ which corresponds to the 
transformation $\lambda\to\lambda^{\vee}$, one finds \eqref{Grothen-B2}
is valid if \eqref{Grothen-B1} holds. Thus, it is sufficient to show \eqref{Grothen-B1}.

The relation between the wavefunctions of the integrable five-vertex model
of $N$ and $N+1$ particles can be reduced to the relation
between the Grothendieck and skew Grothendieck polynomials \eqref{relation}.
This relation is also the key for the non-Hermitian phase model.
Namely, we show the following lemma
for the correspondence between the matrix elements of the
single $\mathcal{B}$- and $\mathcal{C}$- operators
and the skew Grothendieck polynomials of a single variable
from which one concludes that the
wavefunctions is proportional to the Grothendieck polynomials.
\begin{lemma}\label{skew-B-lemma}
The matrix elements of the single
$\mathcal{B}$- and $\mathcal{C}$-operators can be
expressed as the skew Grothendieck polynomials of a single variable as
\begin{align}
\bra \{ m \}_{M,N+1}|(v^{-1}-\beta v)^{1-M} \mathcal{B}(v)| \{ n \}_{M,N}
\ket&=G_{\mu/\lambda}(z;\beta), \label{bosonskewone} \\
\bra \{ n \}_{M,N}|(v^{-1}-\beta v)^{1-M} \mathcal{C}(v)| \{ m \}_{M,N+1}
\ket&=G_{\mu^\vee/\lambda^\vee}(z;\beta), \label{bosonskewtwo}
\end{align}
where the Young diagram $\mu=(\mu_1,\mu_2,\dots,\mu_{N+1})$
$(M-1 \ge \mu_1 \ge \cdots \ge \mu_{N+1} \ge 0)$
is parametrized by the configuration $\{ m \}_{M,N+1}=\{m_0,m_1,\dots,m_{M-1} \}$ ($m_0+
\cdots+m_{M-1}=N+1$) as $\mu=((M-1)^{m_{M-1}},\dots,1^{m_1},0^{m_0})$.
The Young diagram
$\mu^\vee=(\mu_1^\vee,\mu_2^\vee,\dots,\mu_{N+1}^\vee)$
$(M-1 \ge \mu_1^\vee \ge \cdots \ge \mu_{N+1}^\vee \ge 0)$
is given by $\mu_j^\vee=M-1-\mu_{N+2-j}$.
\end{lemma}
Here we first end the proof of Theorem \ref{bosontheorem}
by using Lemma~\ref{skew-B-lemma}. The left
hand side of \eqref{Grothen-B1} is decomposed as
\begin{align}
&\bra \{ n \}_{M,N}|\Psi(\{v \}_N) \ket \nn \\
&\quad =\sum_{\{m^{(0)}\},\dots, \{m^{(N-1)}\}}\bra \{ n \}_{M,N}| \prod_{j=1}^N
\left\{\mathcal{B}(v_j)|\{ m^{(N-j)}\}_{M,N-j}\ket \bra \{ m^{(N-j)}\}_{M,N-j}|\right\}|\Omega\ket.
\end{align}
Then applying Lemma~\ref{skew-B-lemma} to the above decomposition and
using Corollary~\ref{decomposition}, one obtains \eqref{Grothen-B1}.
\end{proof}
\noindent
{\it Proof of Lemma~\ref{skew-B-lemma}.} 
Utilizing the graphical description
and an argument similar to Proposition~\ref{skew},
one immediately sees that \eqref{bosonskewtwo} automatically holds 
if \eqref{bosonskewone} holds. Let us show \eqref{bosonskewone}.
From the matrix elements of the $L$-operator, one finds
\begin{align}
\bra \{ m \}_{M,N+1}|\mathcal{B}(u)| \{ n \}_{M,N} \ket=0, \ \ \ 
\mathrm{unless} \ \ \ \{ m \}_{M,N+1} \triangleright \{ n \}_{M,N}.
\end{align}
See Figure~\ref{Interlacing-B2} for a graphical representation.
For $\{ m \}_{M,N+1}$ and $\{ n \}_{M,N}$,
we introduce $\{ p \}_r=\{0 \le p_1<\cdots<p_r \le M-1 \}$ to be the set of all integers
$p$ such that $m_p=n_p+1$, and
$\{ q \}_s=\{0 \le q_1<\cdots<q_s \le M-1 \}$ to be the set of all integers
$q$ such that $m_q+1=n_q$.
When $\{ m \}_{M,N+1}$ and $\{ n \}_{M,N}$ satisfy the
admissible condition $\{ m \}_{M,N+1} \triangleright \{ n \}_{M,N}$,
$\{ p \}_r$ and $\{ q \}_s$ satisfy
$s=r-1$ and $p_k < q_k < p_{k+1}$ ($k=1,\dots,r-1$).
(see Figure~\ref{Interlacing-B1}, for instance).
One calculates the matrix elements of $\mathcal{B}(v)$
using $\{ p \}_r$ and $\{ q \}_{r-1}$ as
\begin{align}
\bra \{m\}_{M,N+1}|\mathcal{B}(v)|\{ n \}_{M,N} \ket 
&={}_{a} \bra 0 | \bra \{m\}_{M,N+1}| \prod_{j=0}^{M-1} \mathcal{L}_{aj}(v)
|1 \ket_{a} |\{ n \}_{M,N} \ket \nonumber \\
&=v^{\sum_{j=1}^r p_j-\sum_{j=1}^{r-1} q_j-r+1}
\prod_{j=1}^r \prod_{k=p_j+1}^{q_j-1}(v^{-1}-\beta v \delta_{n_k\,0}),
\end{align}
where $q_0=-1, \ q_r=M$.
This can be shown by combining the following partial actions:
\begin{align}
&\prod_{l=q_{j-1}+1}^{q_j} \mathcal{L}_{al}(v)| 1 \ket_a \otimes
\left\{ \otimes_{k=q_{j-1}+1}^{q_j}
| n_k \ket_k \right\} \nonumber \\
&\qquad =v^{p_j-q_{j-1}-1} \prod_{l=p_j+1}^{q_j-1}(v^{-1}-\beta v \delta_{n_l\,0})
| 1 \ket_a \left\{ \otimes_{k=q_{j-1}+1}^{q_j}
| m_k \ket_k \right\} \quad (1 \le j \le r-1), \nonumber \\
&\prod_{l=q_{r-1}+1}^{M-1} \mathcal{L}_{al}(v)| 1 \ket_a  \otimes
\left\{ \otimes_{k=q_{r-1}+1}^{M-1}
| n_k \ket_k \right\} \nonumber \\
&\qquad =v^{p_r-q_{r-1}-1} \prod_{l=p_r+1}^{M-1}(v^{-1}-\beta v \delta_{n_l\,0})
| 0 \ket_a \left\{ \otimes_{k=q_{r-1}+1}^{M-1}
| m_k \ket_k \right\}.
\end{align}
Dividing the matrix elements 
$\bra \{m\}_{M,N+1}|\mathcal{B}(v)|\{ n \}_{M,N} \ket$
by $(v^{-1}-\beta v)^{M-1}$
and expressing in terms of the variable $z$, we have
\begin{align}
&\bra \{m\}_{M,N+1}|(v^{-1}-\beta v)^{1-M}
\mathcal{B}(v)|\{ n \}_{M,N} \ket \nonumber \\
&\,\,=
\begin{cases}
z^{\sum_{j=1}^r p_j-\sum_{j=1}^{r-1} q_j}(1+\beta z)^{r-1}
\prod_{j=1}^r \prod_{k=p_j+1}^{q_j-1}(1+\beta z-\beta z \delta_{n_k\,0}) & 
\text{ $\{ m \}_{M,N+1} \triangleright \{ n \}_{M,N}$}  \\
0 & \text{ otherwise}
\end{cases}.
\label{beforetranslationone}
\end{align}
The remaining step is to translate the configuration of particles
$\{ m \}_{M,N+1}$ and $\{ n \}_{M,N}$ with the differences specified by
$\{ p \}_r$ and $\{ q \}_{r-1}$, to
the Young diagrams $\mu$ and $\lambda$.
One finds the translation rule
\begin{align}
&\sum_{j=1}^r p_j-\sum_{j=1}^{r-1}q_j
=\sum_{j=1}^{N+1} \mu_j-\sum_{j=1}^N \lambda_j, \nn \\
&r-1+\# \{k \in \cup_{j=1}^{r} \{ p_j+1, \dots , q_j-1 \} |n_k \neq 0 \}
=\# \{ j \in \{1, \dots , N \}|\lambda_j \neq \mu_{j+1} \}.
\end{align}
By this translation together with Proposition~\ref{admissible}, 
one finds that \eqref{beforetranslationone}
is nothing but the skew Grothendieck polynomial
\eqref{skewexpression}. \hfill $\Box$
\begin{example}
The wavefunctions \eqref{Grothen-B1} and \eqref{Grothen-B2} for
the free phase model ($\beta=0$) reduce to \cite{Bo}:
\begin{align}
\bra \{ n \}_M|\Psi(\{v \}_N) \ket
&=\prod_{j=1}^N v_j^{-M+1} s_\lambda(v_1^2,v_2^2,\dots,v_N^2), \nn \\
\bra \Psi(\{v \}_N)| \{ n \}_M \ket  
&=\prod_{j=1}^N v_j^{M-1} s_\lambda(v_1^{-2},v_2^{-2},\dots,v_N^{-2}),
\end{align}
where $s_{\lambda}(z_1,\dots,z_N)$ is the Schur polynomials.
\end{example}
\begin{example}
For some particular cases, the wavefunctions reduces to some simple forms:
\begin{align}
&\bra \{ N,0,\dots,0 \}|\Psi(\{v \}_N) \ket=
\bra \Psi(\{ v \}_N)|\{0,\dots,0,N \} \ket
=\prod_{j=1}^N (v_j^{-1}-\beta v_j)^{M-1}, \\
&\bra \{ 0,\dots,0,N \}|\Psi(\{v \}_N) \ket=\prod_{j=1}^N v_j^{M-1}=
\bra \Psi(\{ v \}_N)|\{N,0,\dots,0 \} \ket
=\prod_{j=1}^N v_j^{M-1}.
\end{align}
Their relations can be easily checked from
their graphical descriptions.
\end{example}

Applying the Cauchy identity \eqref{cauchy} and using the relations
\eqref{Grothen-B1} and \eqref{Grothen-B2}, we can express the 
scalar product of the $N$-particle states as a determinant
form:
\begin{corollary}
The scalar products of the $N$-particle states for the 
non-Hermitian phase model has the following determinant 
representation.
\begin{align}
&\bra \Psi(\{ u \}_N)|\Psi(\{ v \}_N) \ket \nonumber \\
&=\prod_{1 \le j<k \le N} \frac{1}{(v_j^2-v_k^2)(u_k^2-u_j^2)}
\mathrm{det}_N 
\left[
\frac{(u_k^{-1}-\beta u_k)^M v_j^{M+2(N-1)}
-(v_j^{-1}-\beta v_j)^M u_k^{M+2(N-1)}}{v_j/u_k-u_k/v_j}
\right]. \label{Scalar-B}
\end{align}
\end{corollary}
We can also use the summation formula for the Grothendieck polynomials
\eqref{sum-Gr}
to obtain the summation formula for the wavefunctions of the
non-Hermitian phase model.
\begin{corollary}
The summation formula for the wavefunctions holds.
\begin{align}
&\sum_{\{ n \}_{M,N}}(-\beta)^{\sum_{j=1}^{M-1} jn_j}
\bra \{ n \}_{M,N}|\Psi(\{ v \}_N) \ket
=\frac{\prod_{j=1}^N v_j^{N-1}(v_j^{-1}-\beta v_j)^{M+N-2}}
{\prod_{1\le j<k\le N}(v_k^2-v_j^2)}
\mathrm{det}_N V
\label{bosonsum}
\end{align}
with an $N\times N$ matrix $V$ whose matrix elements are
\begin{align}
&
V_{jk}=\sum_{m=0}^{j-1}(-1)^m (-\beta)^{j-N}
\binom{M+N-1}{m}(1-\beta v_k^2)^{1-m+j-N}
\quad (1\le j \le N-1), \nn \\
&
V_{Nk}=-\sum_{m=\mathrm{max}(N-1,1)}^{M+N-1} (-1)^m
\binom{M+N-1}{m} (1-\beta v_k^2)^{1-m}.
\end{align}
\end{corollary}

\section{Melting crystals}
As an application of our formulae developed in the previous sections,
we study the statistical mechanical system of a melting crystal in 
three dimensions as depicted in Figure~\ref{melting-crystal}.  The melting 
rules are the following.  The melting starts at the one corner of the cubic
crystal. Each cube can be removed if its three faces never touch the other cubes 
constructing
the crystal. The removed cube contributes the factor $q=e^{-\mu/T}$ to the
Boltzmann weight of the configuration, where $\mu>0$ and $T>0$ 
(i.e. $0<q<1$) denote the 
chemical potential and the temperature, respectively.  The system of the melting
crystal can be mapped to the model of stacking cubes around the one corner of
the empty box: a cube can be added such that its three faces touch the other cubes
or the walls/floor of the box. In Figure~\ref{plane-partition}, we depict 
the configuration of the stacked cubes corresponding to Figure~\ref{melting-crystal}. 
One finds that the configurations of the stacked cubes (or equivalently those 
of the melting 
crystal) are in one-to-one correspondence with plane partitions defined as follows.
\begin{figure}[ttt]
\begin{center}
\includegraphics[width=0.4\textwidth]{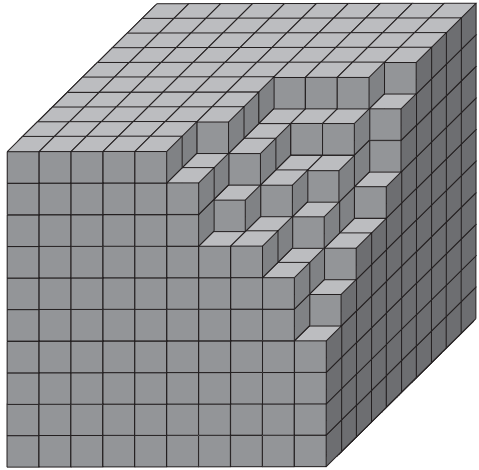}
\end{center}
\caption{A melting crystal. The melting starts at the one corner of 
the crystal. Each cube  is possibly removed (melt) only if its three 
faces do not touch the other cubes. Each removed cube contributes the 
factor $q=e^{-\mu/T}$ ($\mu>0$, $T>0$) 
to the weight of the configuration.} \label{melting-crystal}
\end{figure}

\begin{figure}[ttt]
\begin{center}
\includegraphics[width=0.4\textwidth]{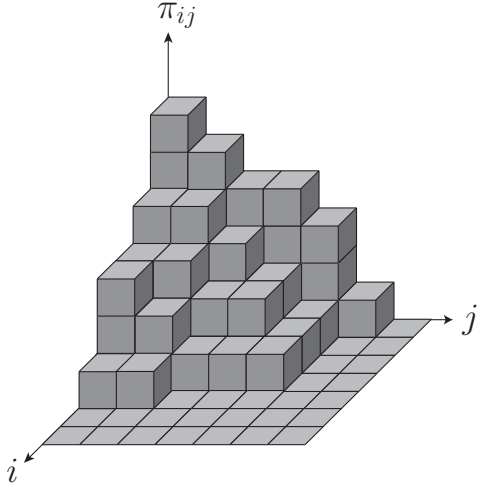}
\end{center}
\caption{The stacked cubes corresponding to Figure~\ref{melting-crystal}. The 
configurations  are in one-to-one correspondence with plane 
partitions. The above configuration
is described by a configuration of the plane partition $|\pi|=58$.}
\label{plane-partition}
\end{figure}
\begin{definition}
A plane partition $\pi$ is a two-dimensional array of non-negative integers 
$\pi_{i j}$ ($i,j>0$) satisfying
$\pi_{i  j} \ge \pi_{i+1 \, j}$, $\pi_{ij} \ge \pi_{i \, j+1}$.
\end{definition}
\noindent
The plane partitions can be regarded as a three-dimensional generalization 
of the Young diagram. In this three-dimensional diagram, $\pi_{ij}$ corresponds
to the height of stacked cubes on the coordinate $(i,j)$. Then
the total number of the stacked cubes is given by 
$|\pi|:=\sum_{i,j\ge 1}\pi_{i j}$.
For later convenience, let us describe some properties satisfying the 
{\it diagonal slices}  of $\pi$, which is defined as follows.
\begin{definition}
For a plane partition $\pi$, the $m$th ($m\in \mathbb{Z}$) 
diagonal slice $\pi^{(m)}$ is a sequence whose 
elements are defined as
\begin{align}
\pi^{(m)}_j=\pi_{j-m\,\, j} \quad \text{for $j>\max(0,m)$}.
\end{align}
\end{definition}
\noindent
First, each diagonal slice $\pi^{(m)}$ is a partition, i.e.,
a sequence of weakly decreasing non-negative integers.
Second, theses partitions satisfy the following
interlacing property.
\begin{lemma}
The series of partitions $\pi^{(m)}$ satisfies the
interlacing relation
\begin{align}
\cdots\pi^{(-2)} \prec \pi^{(-1)}
\prec \pi^{(0)} \succ \pi^{(1)} \succ \pi^{(2)} \cdots.
\label{interlace-lemma}
\end{align}
\end{lemma}
\noindent
See Figure~\ref{diagonal-slice} for an example of the diagonal slices.
\begin{figure}[ttt]
\begin{center}
\includegraphics[width=0.6\textwidth]{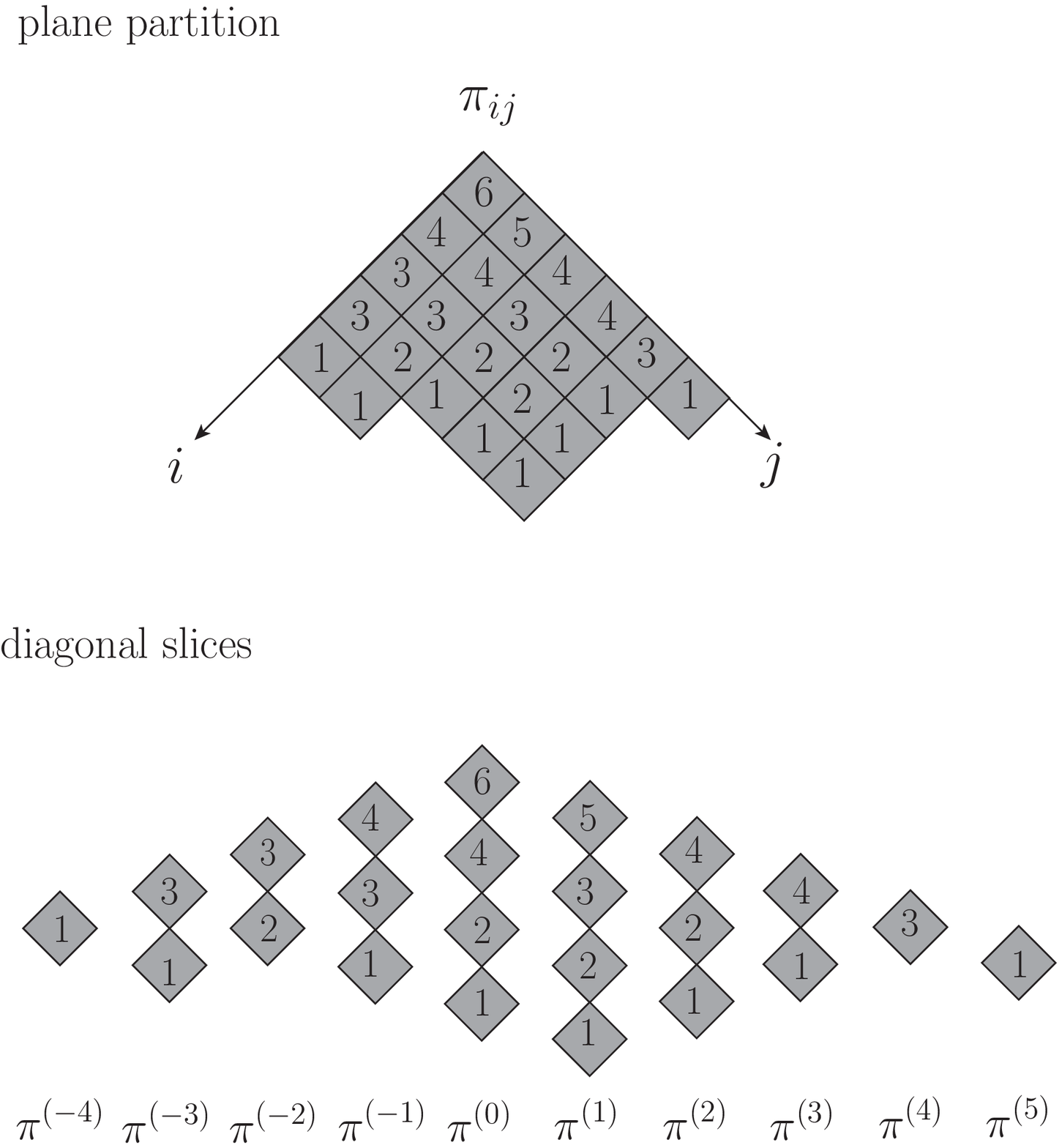}
\end{center}
\caption{The diagonal slices of the plane partition $\pi$ defined
in Figure~\ref{plane-partition}.}
\label{diagonal-slice}
\end{figure}

The partition function $Z$ of the system of the melting crystal is
regarded as the generating function of the plane partition, and
is  known to be given by the so-called MacMahon function \cite{Mac}:
\begin{align}
Z=\sum_{\pi} q^{|\pi|}=\prod_{n=1}^\infty\frac{1}{(1-q^n)^n} \quad  (0<q<1).
\label{MacMahon}
\end{align}

Now we consider the case that a plane partition $\pi$ is contained in a certain finite
box of  size, say $N_1\times N_2 \times L$.  Let us call such a partition the
boxed plane partition and write it 
as $\pi \subseteq [N_1,N_2,L]$.  For this boxed plane partition, the following
is valid:  
\begin{align}
\pi^{(-N_1)}=\pi^{(N_2)}=\emptyset, \quad
\pi_{i\,N_2+1}=\pi_{N_1+1\,j}=0 \,\,\, (i,j \ge 1),
\label{boxed-case1}  
\end{align}
and hence the interlacing
relation \eqref{interlace-lemma} is restricted to
\begin{align}
\emptyset=\pi^{(-N_1)}\prec \cdots \prec \pi^{(-1)}
\prec \pi^{(0)} \succ \pi^{(1)} \succ  \cdots \succ \pi^{(N_2)}=
\emptyset.
\label{boxed-case2}
\end{align}
This case corresponds to a system of the melting
rectangular crystal  of size $N_1\times N_2\times L$.
Then the  partition function of the system is given by \cite{Mac}:
\begin{align}
Z_{\rm box}=\sum_{\pi \subseteq [N_1,N_2,L]} q^{|\pi|}=
\prod_{j=1}^{N_1}\prod_{k=1}^{N_2} \frac{1-q^{L+j+k-1}}{1-q^{j+k-1}}.
\label{partition-func}
\end{align}
In the limit $q\to 1$, this formula gives the number of the plane partitions
contained in the box $N_1\times N_2\times L$.
In \cite{Bo}, the formula \eqref{partition-func} for the box $N\times N \times L$
is reproduced by utilizing the scalar products of the phase model.  

Inspired by that work, we extend the method to the 
case for the non-Hermitian phase model and calculate the 
partition function for the statistical mechanical model of a melting crystal
with the size $N\times N \times L$. The partition function of the model is
defined as \footnote{We remark again as in the introduction that this
assignment of the weights for each plane partition is totally different
from the ones in previous literature like \cite{Vu,FW} for example,
which are based on the Macdonald polynomials and its degeneration to the
Hall-Littlewood polynomials.
The model introduced in this paper is based on the Grothendieck polynomials,
and the directions of the extensions from the Schur to
the Grothendieck and the Macdonald polynomials are different,
hence are the corresponding melting crystal models.
}
\begin{align}
&Z_{\rm box}(\beta)=\sum_{\pi \subseteq [N,N,L]}  \Phi(q,\beta;\pi) q^{|\pi|} \quad
(0< q <1),\nn \\
&\Phi(q,\beta;\pi)=\prod_{j=1}^{N}\prod_{k=1}^{N-j}\left[
             (1+\beta q^j)^{-\delta(\pi_{k}^{(j)}, \,\pi_{k+1}^{(j-1)})}
             (1+\beta q^{1-j})^{1-\delta(\pi_{k}^{(-j)}, \,\pi_{k}^{(1-j)})}\right],
\label{part-def}
\end{align}
where $\delta(i,j)$ denotes the Kronecker delta: $\delta(i,j)=\delta_{i\,j}$.
Here we comment on the physical meaning of the additional 
potential factor $\Phi(q,\beta;\pi)$.  This factor can be 
interpreted to reflect, such as microscopic interactions 
among atoms. For $\beta>0$, it brings out a surface flattening 
effect in the region $j>i$ in Figure~\ref{plane-partition} or 
\ref{diagonal-slice}. In contrast to this, in the region $j<i$, 
the potential causes a surface roughening effect. The strength of the
effects decreases (resp. increases) with distance from the
plane $i=j$ in the region $j> i$ (resp. $i<j$).

On the other hand, for $\beta<0$, the potential structure is 
much more complicated. (i) For $\beta<-2$, the potential 
$\Phi(q,\beta;\pi)$ denotes a roughening effect in 
$i<j<\log(-2/\beta)/\log(q) \cup j<i$,
and a flattening effect in the other region. (ii) For $-2\le\beta<0$, 
it denotes a roughening effect in $1-\log(-2/\beta)/\log q<j<i$, and 
a flattening effect in the other region.
Note that for $\beta<0$, the model sometimes
becomes {\it physically} ill-defined, because $\Phi(q,\beta;\pi)$
possibly takes negative values.

In any cases, due to the strength of the force is not symmetric
with respect to the plane $i=j$, the expected shape of the
melting crystal is not symmetric with respect to $i=j$ except
for $\beta=0$.

The partition function $Z_{\rm box}(\beta)$ is explicitly evaluated 
by using the Cauchy identity \eqref{cauchy} and
Corollary~\ref{decompositionskew}. The following and subsequent
Corollaries are 
the main results of this section.
\begin{corollary}
The partition function $Z_{\rm box}(\beta)$ is given by 
\begin{align}
Z_{\rm box}(\beta)&=\frac{q^{N(N-1)/2}\prod_{j=1}^N (1+\beta q^j)^{j-1}}
                      {\prod_{1\le j<k\le N} (q^j-q^k)^2}
             \mathrm{det}_N
                \left[\frac{1-q^{(j+k-1)(L+N)}
                 \left(\frac{1+\beta q^{1-k}}{1+\beta q^j}\right)^{N-1}}
                     {1-q^{j+k-1}}\right].
\label{partition}
\end{align}
\end{corollary}
\begin{proof}
Consider the Cauchy identity given by \eqref{cauchy} for $\lambda=\pi^{(0)}$.
Then applying Corollary~\ref{decompositionskew}, the Grothendieck polynomials
in the left hand side are given by
\begin{align}
&G_{\pi^{(0)}}(z_1,\dots,z_N;\beta)=
\sum_{\pi^{(0)}\succ\dots \succ \pi^{(N)}=\emptyset}
\prod_{l=1}^N z_l^{|\pi^{(l-1)}|-|\pi^{(l)}|} 
\prod_{j=1}^N\prod_{k=1}^{N-j}
\left[1+\beta z_j-\beta z_j \delta_{\pi_{k+1}^{(j-1)}\,\pi_k^{(j)}}\right], \nn \\
&G_{\pi^{(0)\vee}}(w_1,\dots,w_N;\beta)=
\sum_{\emptyset=\pi^{(-N)\vee}\prec\dots \prec \pi^{(0)\vee}}
\prod_{l=1}^N w_l^{L+|\pi^{(-l)}|-|\pi^{(1-l)}|}
 \nn \\
& \qquad \qquad \qquad \qquad \qquad \quad 
\times \prod_{j=1}^N\prod_{k=1}^{N-j}
\left[1+\beta w_j-\beta w_j \delta_{\pi_{k}^{(1-j)}\,\pi_k^{(-j)}}\right].
\end{align}
Here we have used $\pi^{(j)\vee}_k=L-\pi^{(j)}_{N-|j|-k+1}$ and
the properties \eqref{correspond}, \eqref{boxed-case1} and 
\eqref{boxed-case2} for the explicit
evaluations. The insertion of them into the 
Cauchy identity \eqref{cauchy} yields
\begin{align}
&\sum_{\pi \subseteq [N,N,L]}
\prod_{j=1}^N z_j^{|\pi^{(j-1)}|-|\pi^{(j)}|} 
w_j^{|\pi^{(-j)}|-|\pi^{(1-j)}|}  \nn \\
&\quad \qquad \quad\times
\prod_{j=1}^{N}\prod_{k=1}^{N-j}\left[
             (1+\beta z_j)^{-\delta(\pi_{k}^{(j)}\, ,\pi_{k+1}^{(j-1)})}
             (1+\beta w_j)^{1-\delta(\pi_{k}^{(-j)}\, , \pi_{k}^{(1-j)})}\right] \nn \\
&\quad=  
\frac{\prod_{j=1}^{N}(1+\beta z_j)^{j-1}}
{\prod_{1\le j<k\le N}(z_j-z_k)(w_j^{-1}-w_k^{-1})}
\mathrm{det}_N\left[\frac{
1-(z_j w_k^{-1})^{L+N}\left(\frac{1+\beta w_k}{1+\beta z_j}\right)^{N-1}}
{1-z_j w_k^{-1}}\right].
\end{align}
Setting $z_j=q^{j}$ and $w_j=q^{1-j}$ in the above, we finally arrive at
\eqref{partition-func}.
\end{proof}

Set $\beta=0$ in \eqref{partition}, then the formula \eqref{partition-func} is 
reproduced. Moreover taking the limit $L\to\infty$ and  $N\to\infty$
we have the following generalized MacMahon function which reduces to
the ordinary MacMahon function \eqref{MacMahon} for $\beta=0$ and 
Euler's generating function at $\beta=-1$.
\begin{corollary}
The partition function \eqref{partition} in the limit $L\to \infty$ and
$N\to\infty$ is given by
\begin{align}
Z(\beta):=\lim_{L,N\to\infty}Z_{\rm box}(\beta)=
\prod_{n=1}^\infty\frac{(1+\beta q^n)^{n-1}}{(1-q^n)^n},
\label{unification}
\end{align}
which becomes the MacMahon function and  Euler's generating function
at $\beta=0$ and $\beta=-1$ respectively.
\end{corollary}
\noindent
For $\beta=0$,
the partition function $Z(0)$ is nothing but the MacMahon function
\eqref{partition-func} which is a generating function of the plane partitions.
And surprisingly, for $\beta=-1$ which corresponds to the
TASEP (resp. TAZRP) in the
language of the five-vertex model (resp. the non-Hermitian phase model), 
$Z(-1)$ is nothing but a
generating function for the numbers of possible partitions of 
natural numbers which is due to Euler:
\begin{align}
Z(-1)=\prod_{n=1}^\infty\frac{1}{1-q^n}=\sum_{\lambda}q^{|\lambda|}.
\end{align}
The expression for the partition function \eqref{unification}
means that the melting crystal model we introduced
unifies the generating functions of the
two-dimensional and three-dimensional Young diagrams
\eqref{unification}.
The enumeration problems for
two-dimensional Young diagrams can be treated
by the three-dimensional melting crystal model
at the point $\beta=-1$.
Note that there are other melting crystal models
based on the Macdonald polynomials \cite{Vu,FW}, whose partition functions
are different but have simple expressions in the infinite volume limit
as ours,
including the MacMahon function at a special point.
But if one wants to relate the results in \cite{Vu,FW}
with the Euler's generating function,
one has to multiply infinite products.
This means that one should multiply infinite products to the weights
assigned to each plane partition,
which seems to be an artificial operation,
unnatural from the point of view of enumeration.

Note here that the partition function \eqref{unification} is physically well-defined 
for $\beta \ge -1$ which is a condition for positivity of $Z(\beta)$. 
The entropy $S(\beta)$ for the model \eqref{unification} can be calculated by
using the relation  $Z(\beta)=e^{S(\beta)-E/T}$, where
 $E:=T^2 \der^2 \log Z(\beta)/\der T^2$ is the internal energy. Explicitly it reads
\begin{align}
S(\beta)=\sum_{n=1}^\infty \frac{\mu n}{T}\left[\frac{\beta(n-1)}{\beta+q^{-n}}
               +\frac{n}{q^{-n}-1} \right]+
               \sum_{n=1}^\infty \log\left[\frac{(1+\beta q^n)^{n-1}}{(1-q^n)^n}\right]
\quad (\beta\ge -1),
\label{entropy}
\end{align}
where $q=e^{-\mu/T}$ $(\mu>0, T>0)$.  From this expression, it can be easily followed that
the entropy $S(\beta)$ is a monotonically increasing function of $\beta$. In 
Figure~\ref{entropy-fig},
the temperature dependence of the entropy is depicted for various values of $\beta$.
\begin{figure}[ttt]
\begin{center}
\includegraphics[width=0.7\textwidth]{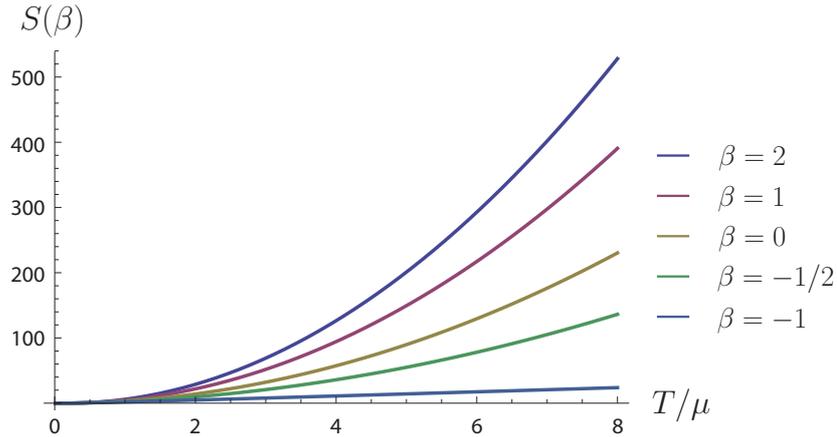}
\end{center}
\caption{The temperature dependence of the entropy $S(\beta)$ \eqref{entropy}
is depicted for various values of $\beta$.}
\label{entropy-fig}
\end{figure}
%
\section{Conclusion}
%
In this paper, we studied the non-Hermitian phase model
and showed that the wavefunctions is nothing but the Grothendieck polynomials.
To show this, we reviewed the integrable five-vertex model,
and introduced the skew Grothendieck polynomials for a single variable
as matrix elements of a $B$-operator. The addition theorem for the
Grothendieck polynomials follows from the equivalence between 
the wavefunctions of the five-vertex model and the Grothendieck 
polynomials. Showing that the matrix element of the $B$-operator
in the non-Hermitian phase model is given by the skew Grothendieck
polynomials, and then applying the addition theorem, we derive the
wavefunctions of the non-Hermitian phase model, which 
can also be expressed by the Grothendieck polynomials.
Our works establish the $K$-theoretic boson-fermion correspondence
at the level of wavefunctions.

As another application of the boson-fermion correspondence, 
we discussed the statistical mechanical model of a three-dimensional 
melting crystal and exactly derive the partition functions, which 
is interpreted as a $K$-theoretic generalization of the MacMahon
function.
Surprisingly, the $K$-theoretic MacMahon function includes not only the
generating function of the plane partitions but also
Euler's generating function of the partitions.
Our refinement of the melting crystal  model
unifies the treatment of the enumeration problems
of two-dimensional and three-dimensional Young diagrams.
The reason why two-dimensional objects appear for $K$-theory is not known now,
and its geometric meaning deserves to be investigated in the future.

The hermitian phase model is described by the
Schur polynomials. Since the determinant representations of the scalar products
are essentially the Cauchy identity for the Schur polynomials,
it has connections with KP equation and the Toda lattice \cite{FW,Zu,Ta}.
It is interesting to examine whether this classical integrable
interpretation can be lifted to the case of the integrable five-vertex model
and the non-Hermitian phase model, by making connection with
the existing classical integrable system or extending to some extent.

In the words of geometry,
our works on the relation between non-Hermitian integrable models
and Grothendieck polynomials mean that non-Hermitian integrable models
provide a natural framework to study the
quantum $K$-theory of Grassmannian varieties.
For the hermitian phase model, the quantum cohomology ring
and the Verlinde ring are shown to be described by the ring
defined by the model under the quasiperiodic boundary condition
\cite{KS}, where the Bethe ansatz equation plays the role of the ideal.
We would like to make further investigations on quantum $K$-theoretic objects
in our framework in the future.

One of the problems we are planning to investigate is to lift
the relation between integrable models and $K$-theoretic objects
to other types of Grassmannian varieties.
There are several extensions and variations of the Schur polynomials.
The Schur $P$, Schur $Q$,
Jack, Hall-Littlewood and the Macdonald polynomials
have connections with the $q$-boson model \cite{Wh,Ts,Ko}.
On the other hand, the $K$-theoretic extension of the Schur $P$ and Schur $Q$
polynomials are introduced in \cite{IN}.
We expect to find connections between these $K$-theoretical
symmetric polynomials and the integrable models
such as the non-Hermitian $q$-boson model \cite{SW,BC,BCPS}, for example.

\section*{Acknowledgments}
We thank C. Arita, T. Ikeda, S. Kakei, A. Kuniba, S. Naito, H. Naruse
and Y. Takeyama for useful discussions.
The present work was partially supported
by Grants-in-Aid  for Scientific Research (C) No. 24540393
and for Young Scientists (B) No. 25800223.

\section*{Appendix}
In this Appendix, we show the $L$-operator
of the five-vertex model \eqref{loperator}
is a particular reduction of a
more general six-vertex model.
We start from the $R$-matrix of the following six-vertex model
satisfying the Yang-Baxter relation \eqref{YBE}
\begin{align}
R(u,v)
&=
\begin{pmatrix}
f(v,u;t) & 0 & 0 & 0 \\
0 & t & g(v,u;t) & 0 \\
0 & g(v,u;t) & 1 & 0 \\
0 & 0 & 0 & f(v,u;t)
\end{pmatrix}, \nonumber \\
f(v,u;t)&=\frac{u^2-t v^2}{u^2-v^2},  \,
g(v,u;t)=\frac{(1-t)uv}{u^2-v^2}, 
\nonumber
\end{align}
including the $R$-matrix of the five-vertex model
\eqref{Rmatrix} as a special point $t=0$.
One can show that the following $L$-operator solves the
$RLL$ relation \eqref{RLL} for this $R$-matrix
of the six-vertex model
\begin{align}
L(u)
=
\begin{pmatrix}
\alpha_3 u+\alpha_4 u^{-1} & 0 & 0 & 0 \\
0 & \alpha_3 t u+\alpha_4 u^{-1} & (1-t)\alpha_1 & 0 \\
0 & (1-t)\alpha_2 & \alpha_5 u+\alpha_6 u^{-1} & 0 \\
0 & 0 & 0 & \alpha_5 u+\alpha_6 t u
\end{pmatrix},
\nonumber
\end{align}
where the parameters $\alpha_j, \ j=1,\cdots,6$ and $t$
satisfy the relations
\begin{align}
&(1-t)\alpha_1 \alpha_2+\alpha_3 \alpha_6-\alpha_4 \alpha_5=0,
\nonumber \\
&(t^2-t)\alpha_1 \alpha_2+t^2 \alpha_3 \alpha_6-\alpha_4 \alpha_5=0.
\nonumber
\end{align}
The $R$-matrix of the six-vertex model is recovered
from the $L$-operator by the choice of the parameters
$\alpha_1=\alpha_2=\alpha_3=\alpha_5=1, \ \alpha_6=-1, \ \alpha_4=-t$.

Another particular choice of the $L$-operator
of the general six-vertex model
$t=\alpha_4=0, \ \alpha_1=\alpha_2=\alpha_3=1, \ \alpha_6=-1, \
\alpha_5=-\beta^{-1}$ gives the $L$-operator for the
five-vertex model \eqref{loperator}, whose wavefunction
are the Grothendieck polynomials.

\end{document}